\documentclass[a4paper,runningheads]{llncs}

\usepackage{etoolbox}
\usepackage[paperheight=235mm, paperwidth=155mm,textwidth=12.2cm,textheight=19.3cm,hmarginratio=1:1]{geometry}

\usepackage{bm}
\usepackage{amsmath,amssymb}
\usepackage[T1]{fontenc}
\usepackage{xspace}
\usepackage{booktabs,multirow}
\usepackage{cite}
\usepackage{wrapfig}
\usepackage{wrapstuff}
\usepackage{graphicx}
\usepackage{xcolor}
\usepackage[%
  bookmarks,
  unicode,
  colorlinks=true,
  allcolors=blue!65!black!90,
  breaklinks=true]{hyperref}
\usepackage{bbding}
\usepackage[caption=false,font=scriptsize]{subfig}
\usepackage{array}

\usepackage{enumitem}
\setitemize{noitemsep,topsep=0pt,parsep=0pt,partopsep=0pt,leftmargin=12.0pt}
\setenumerate{noitemsep,topsep=0pt,parsep=0pt,partopsep=0pt}
\setdescription{noitemsep,topsep=2pt,parsep=0pt,partopsep=0pt,leftmargin=12.5pt}

\usepackage[noend]{algorithmic}

\usepackage{tikz}
\usepackage{pgfplots}
\usetikzlibrary{positioning,patterns}
\tikzstyle{empty} = [fill,circle,inner sep=0mm,minimum size=0pt,line width=0mm]
\tikzstyle{dot} = [fill,circle,inner sep=0mm,minimum size=3pt,line width=0mm]
\tikzstyle{null} = [circle,inner sep=0mm,minimum size=0pt,line width=0mm]
\tikzstyle{state} = [draw,fill=white,rectangle,rounded corners,minimum size=4mm,inner sep=1pt,thick]
\tikzstyle{split} = [draw,circle, fill=black,scale=0.4]
\pgfplotsset{compat=1.18}

\pgfplotsset{
    discard if not/.style 2 args={
        x filter/.code={
            \edef\tempa{\thisrow{#1}}
            \edef\tempb{#2}
            \ifx\tempa\tempb
            \else
                
            \fi
        }
    }
}

\makeatletter
\def\orcidID#1{\textsuperscript{\,\smash{\protect\raisebox{-1.25pt}{\href{http://orcid.org/#1}{\protect\includegraphics[scale=.8]{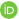}}}}}}
\makeatother

\makeatletter%
\g@addto@macro\normalsize{%
  \setlength\abovedisplayskip{3pt}%
  \setlength\belowdisplayskip{3pt}%
  \setlength\abovedisplayshortskip{-3pt}%
  \setlength\belowdisplayshortskip{3pt}%
}%
\makeatother
\allowdisplaybreaks %

\usepackage[nameinlink,capitalize,english]{cleveref}
\Crefname{figure}{Fig.}{Figs.}
\crefname{figure}{fig.}{figs.}
\Crefname{tabular}{Tab.}{Tabs.}
\crefname{tabular}{tab.}{tabs.}
\Crefname{section}{Sect.}{Sects.}
\crefname{section}{sect.}{sects.}
\Crefname{appendix}{App.}{Apps.}
\crefname{appendix}{app.}{apps.}
\Crefname{equation}{Eq.}{Eqs.}
\crefname{equation}{eq.}{eqs.}
\creflabelformat{equation}{#2#1#3}
\Crefname{example}{Ex.}{Exs.}
\crefname{example}{ex.}{exs.}
\Crefname{definition}{Def.}{Defs.}
\crefname{definition}{def.}{defs.}
\Crefname{algorithm}{Alg.}{Algs.}
\crefname{algorithm}{alg.}{algs.}

\makeatletter
\AddToHook{cmd/appendix/before}{\def\cref@section@alias{appendix}}
\makeatother

\usepackage[vlined,linesnumbered]{algorithm2e}

\SetAlgoCaptionLayout{raggedright}
\SetAlCapFnt{\small}
\SetAlCapNameFnt{\small}
\SetAlCapSkip{\abovecaptionskip\relax}
\SetEndCharOfAlgoLine{}
\SetArgSty{textrm}
\SetKwInOut{Input}{Input}\SetKwInOut{Output}{Output}
\SetCommentSty{textit}
\SetKw{Break}{break}
\SetKwProg{Type}{type}{}{end}
\SetKwProg{Function}{function}{}{end}
\SetKwFor{ForEach}{foreach}{do}{}
\SetKwFor{WhileTrue}{repeat}{}{end}
\SetKw{Continue}{continue}
\SetKwRepeat{DoWhile}{do}{while}
\SetKwProg{Fn}{function}{}{}
\crefname{algocf}{alg.}{algs.}
\Crefname{algocf}{Alg.}{Algs.}

\newcommand{\paragraphskip}{\vspace{12pt plus 4pt minus 4pt}} %
\renewcommand{\emptyset}{\varnothing}
\newcommand{\RR}{\ensuremath{\mathbb{R}}\xspace}  %

\newcommand{\tool}[1]{\textsc{#1}}

\newcommand{\toolset}{\tool{Modest Toolset}\xspace}
\newcommand{\modes}{\tool{modes}\xspace}
\newcommand{\mcsta}{\tool{mcsta}\xspace}
\newcommand{\eg}{e.g.\ }
\newcommand{\ie}{i.e.\ }

\newcommand{\wrt}{w.r.t.\xspace}

\newcommand{\set}[1]{\ensuremath{\{\,#1\,\}}}
\newcommand{\tuple}[1]{\ensuremath{\langle #1 \rangle}}
\newcommand{\powerset}[1]{\ensuremath{2^{#1}}\xspace}
\newcommand{\defeq}{\mathrel{\vbox{\offinterlineskip\ialign{\hfil##\hfil\cr{\tiny \rm def}\cr\noalign{\kern0.30ex}$=$\cr}}}}

\newcommand{\selcandidates}{\mathsf{select\_candidates}}
\newcommand{\est}[2]{\hat{x}_{#1}^{#2}}
\newcommand{\precision}[2]{\varepsilon_{#1}^{#2}}
\newcommand{\Ach}[1]{\ensuremath{\mathit{Ach}(#1)}}

\pgfmathdeclarefunction{gauss}{2}{\pgfmathparse{1/(#2*sqrt(2*pi))*exp(-((x-#1)^2)/(2*#2^2))}}

\begin{document}

\title{%
Multi-Objective Statistical Model Checking\\ using Lightweight Strategy Sampling%
\thanks{
Authors are sorted in alphabetical order.
This work was supported
by the EU's H2020 R\,\&\,I programme under MSCA grant agreement 101008233 (MISSION),
by the Interreg North Sea project STORM\_SAFE,
by SeCyT-UNC grant 33620230100384CB (MECANO),
and
by NWO VIDI grant VI.Vidi.223.110 (TruSTy).
}\\ (extended version)
}
\titlerunning{Multi-Objective Statistical Model Checking using LSS}

\author{%
Pedro R.\ D'Argenio\inst{1}\orcidID{0000-0002-8528-9215}
\and
Arnd Hartmanns\inst{2}\orcidID{0000-0003-3268-8674}
\and\\
Patrick Wienh\"oft\inst{3,4}\orcidID{0000-0001-8047-4094}
\and
Mark van Wijk\inst{2}$^{\text{\,(\raisebox{-1.6pt}{\Envelope})}}$\orcidID{0009-0001-8239-3164}
}
\authorrunning{P.\ R.\ D'Argenio, A.\ Hartmanns, P.\ Wienh\"oft, M.\ van Wijk}
\institute{%
Universidad Nacional de Córdoba and CONICET, Córdoba, Argentina
\and
University of Twente, Enschede, The Netherlands
$\cdot$ \email{mark.vanwijk@utwente.nl}
\and
Dresden University of Technology, Dresden, Germany
\and
Centre for Tactile Internet with Human-in-the-Loop (CeTI), Dresden, Germany
}

\maketitle

\begin{abstract}
Statistical model checking delivers quantitative verification results with statistical guarantees. %
It scales to model sizes and model types that are out of reach for exhaustive, analytical techniques.
So far, it has been used to evaluate one property value at a time only.
Many practical problems, however, require finding the Pareto front of optimal tradeoffs between multiple objectives.
In this paper, we present the first statistical model checking approach for such multi-objective Pareto queries, based on lightweight strategy sampling. %
We introduce an incremental scheme that almost surely converges to a statistically sound confidence band around the true Pareto front in the long run.
To obtain a close underapproximation of the true front in finite time, we propose three heuristic approaches that try to make the best of an a-priori fixed sampling budget.
We implement our new techniques in the \modes simulator, and show their effectiveness on benchmarks from the literature.
\end{abstract}

\section{Introduction}
\label{sec:Intro}

Statistical model checking (SMC)~\cite{AP18,YS02,LLTYSG19} is %
a popular alternative to probabilistic model checking (PMC)~\cite{BAFK18,Bai16} for applications where the latter runs into the state space explosion problem due to its exhaustive nature, or where no effective PMC algorithm is available for the model class at hand. %
SMC circumvents these limitations by applying Monte Carlo simulation: %
Given an effectively executable probabilistic model~$M$ %
and property of interest~$\phi$, it performs random executions (or \emph{simulation runs}) of $M$ to obtain samples for the satisfaction or numeric value of~$\phi$.
Given a desired confidence level, SMC then computes a confidence interval around the satisfaction probability or expected value of~$\phi$,
providing a statistical correctness guarantee.
The core types of properties handled by PMC and SMC are reachability probabilities and expected accumulated rewards~\cite{HJQW23,BHMWW25a}, and both can be extended to compute probabilities for $\omega$-regular objectives in \eg LTL~\cite{Pnu77}~\cite{Kat16,ADKW20}.
PMC and more recently SMC~\cite{BHMWW25b} have also been extended to other quantities like quantiles~\cite{UB13,KSBD15} and conditional value at risk~\cite{KM18}. %

In practice, we often face tradeoffs between multiple properties, such as maximising a satellite's expected utility without %
depleting its battery~\cite{BGHK+19}, or balancing cost and comfort in heat-pump control~\cite{HJLS23}.
That is, for a probabilistic model with controllable nondeterministic choices, we want to choose among the set of \emph{Pareto-optimal} control strategies:
those where no change can be made that improves the value for at least one property without worsening any other.
While PMC algorithms for such \emph{multi-objective} problems are more than a decade old~\cite{FKNPQ11,FKP12} and cover various models~\cite{Qua23} and properties~\cite{RRS17}, we are not aware of any work that brings the advantages of SMC to the multi-objective setting.

In this paper, we present the \textbf{first multi-objective SMC approach} for Pareto queries.
Its \textbf{memory usage is constant} in the size of the model's state space, preserving the distinguishing advantage of SMC.
It uses lightweight strategy sampling (LSS)~\cite{LST14} to evaluate many randomly chosen control strategies \wrt all objectives at hand, and from the resulting data constructs a \textbf{statistically sound} approximation of the (unknown) true Pareto front.

We first propose an incremental scheme that keeps sampling batches of new strategies, continuously refining an under- and overapproximation of the Pareto front.
Assuming ideal LSS and taking care to properly handle statistical error accumulation, the two approximations almost surely converge to a \emph{simultaneous confidence band} for the true Pareto front \emph{in the long run}.
Second, we introduce three fixed-budget algorithms aiming to find the best approximation of the Pareto front in \emph{finite time}. %
They deliver statistically sound \emph{simultaneous lower bounds}, %
since obtaining guarantees on upper bounds is impossible in finite time.

We implement our approach in the \toolset~\cite{HH14}'s \modes simulator~\cite{BDHS20} to experimentally evaluate the effectiveness of SMC-based multi-objective model checking and compare the three fixed-budget algorithms, using benchmarks from the quantitative verification and reinforcement learning (RL)~\cite{SB18} literature. %
Our SMC approaches deliver nontrivial fronts for instances too large (in terms of state space size and dimensionality of the multi-objective query) to be checked by \tool{Storm}'s~\cite{HJKQV22} multi-objective PMC engine.
We compare the quality of our approximation with the Pareto fronts computed by \tool{Storm} where possible.
We use Markov decision process (MDP)~\cite{How60,Bel57,Put94} models in this paper, but our work directly extends to any model class supported by LSS, \eg Markov automata~\cite{EHZ10,BDH24} and probabilistic timed automata~\cite{KNSS02,DHLS16,HSD17}.

\paragraph{Related work.}
Multi-objective PMC uses linear programming~\cite{FKNPQ11,EKVY08} or value iteration~\cite{FKP12,Qua23}.
It is implemented in \tool{Storm}, \tool{Prism}~\cite{KNP11}, and \tool{ePMC}~\cite{FHLSSTZ22}; see \cite[Section~6]{ABBC+24} for a summary of the state of the art.
While PMC can provide hard (non-statistical) guarantees, it requires exploration of the model's state space.
SMC, in contrast, samples instead of exploring, avoiding the state space explosion problem and easily handling models with billions of states~\cite{LL16}.

\tool{Uppaal Smc}~\cite{DLLMP15} %
can compare the values of two properties via SMC on a fully stochastic model.
\tool{Uppaal Stratego}~\cite{DJLMT15} brings together SMC and nondeterminism:
It first finds a \emph{most permissive} strategy for a set of hard constraints via real-time model checking~\cite{BCDF+07}, %
then among the choices allowed by this strategy, uses RL to optimise a ``soft'' quantity~\cite{DJLL+14}, %
which is finally evaluated via SMC.
\tool{Stratego} thus solves a \emph{lexicographic objectives} problem~\cite{WZM15,HPSSTW21} and performs one-dimensional quantitative optimisation only.
We, in contrast, determine the trade-offs between multiple equal-priority probabilistic properties.
We do so in constant memory, while both \tool{Stratego}'s initial model checking and RL store per-state information, potentially limiting scalability.
The latter equally applies to other RL-based multi-objective analysis
approaches~\cite{ACOD16,HRBK+22}.

The general field of \emph{multi-objective optimisation}~\cite{SK06,NRR14} uses many different techniques such as evolutionary algorithms~\cite{DAPM02,SLN02,VMTD15}, particle swarm optimisation~\cite{PV02}, or ant colony optimisation~\cite{AMAD+25}.
These are typically applied to functions of real numbers but can be extended to MDPs~\cite{WCL19,ZLXYW22}.
Like our methods, in finite time, these approaches provide lower bounds only.
They however require more than constant memory and typically do not offer statistical guarantees.

\section{Background}
\label{sec:Background}

\begin{definition}
\label{def:MDP}
A \emph{Markov decision process} (MDP) is a tuple
$\mathcal{M} = \tuple{\mathcal{S}, s_{\mathit{init}}, \mathcal{A}, \delta}$
consisting of finite sets of states $\mathcal{S}$ (its state space) and actions $\mathcal{A}$,
an initial state $s_{\mathit{init}} \in \mathcal{S}$,
and
a partial transition probability function $\delta \colon \mathcal{S} \times \mathcal{A} \rightharpoonup (\mathcal{S} \to [0, 1])$.
Define $\mathcal{A}(s)$ as the set of all actions for which $\delta(s, a)$ is defined.
For all $s \in \mathcal{S}$, we require $|\mathcal{A}(s)| \geq 1$ and $\forall a \in \mathcal{A}(s)\colon \!\sum_{s' \in \mathcal{S}}\delta(s, a)(s') = 1$.
A \emph{reward structure} for an MDP is a function $\mathcal{R}\colon \mathcal{S} \times \mathcal{A} \times \mathcal{S} \to \RR$.
\end{definition}
A reward structure provides a \emph{reward} when moving from one state to another.
We call $\delta(s, a)$ for $a \in \mathcal{A}(s)$ a \emph{transition}; its \emph{branches} are $\set{ s' \mid \delta(s, a)(s') > 0 }$.

\begin{example}
MDP $M_R$ shown in \Cref{fig:ExampleMDP} %
models the process of writing a paper (action $\mathit{write}$) or not ($\mathit{stop}$), and if successfully written (probability $0.85$), choosing to $\mathit{subm}$it it to a conference or to $\mathit{arch}$ive it on arXiv.
The conference's acceptance rate is $20\,\%$; if rejected, we can try again the next year or go to arXiv.
Reward structure $\mathcal{R}_\mathit{rec}$ represents the recognition our results get, and $\mathcal{R}_\mathit{eff}$ our effort.
We annotate branch $s'$ of $\delta(s, a)$ with $(\mathcal{R}_\mathit{rec}(s, a, s') \,{/}\, \mathcal{R}_\mathit{eff}(s, a, s'))$, omitting $(0/0)$s and writing rewards that are the same for all branches on the transition instead.
\end{example}
A path is a sequence $s_0\, a_0\, \ldots \in (\mathcal{S} \times \mathcal{A})^\omega$ s.t.\ $\forall i\colon a_i \in \mathcal{A}(s_i) \wedge \delta(s, a)(s_{i+1}) > 0$.
In an MDP, the choice of action in each state is nondeterministic.
A \emph{probabilistic memoryless strategy} $\sigma \colon \mathcal{S} \rightarrow (\mathcal{A} \to [0,1])$ with $\sum_{a \in \mathcal{A}(s)}\sigma(s)(a) = 1$ makes this choice in all states~$s$.
If $\forall s\,\exists\,a\colon \sigma(s)(a) = 1$, the strategy is \textit{deterministic memoryless} (DM).
When removing all state-actions not chosen by $\sigma$, %
we get the induced discrete-time Markov chain $M|_\sigma$,
on which we can construct probability measures $\mathbb{P}_\sigma^{M,s}$ on the measurable sets of paths starting in~$s$~\cite[Sect.~2.2]{FKNP11}. %

\begin{figure}[t]
\centering
  \subfloat[Example MDP $M_R$]{
    \centering
    \begin{tikzpicture}
      \node[state] (i) {\strut$\,\mathit{init}\,$};
      \node[state] (p) [above right=0.2 and 0.75 of i] {\strut$\,\mathit{paper}\,$};
      \node[state] (f) [below=1.7 of p] {\strut$\,\mathit{done}\,$};

      \node (init) [above=0.3 of i] {};
      \draw[->] (init) to (i);

      \node[split] (is1) [below left=0.6 and -0.2 of p] {};
      \draw[-,shorten <=-2pt] (i) to[bend right=25] node[above,inner sep=4pt] {$\mathit{write}$} node[below,inner sep=2pt] {\scriptsize(0/\textbf{\texttt{+}10})} (is1);
      \draw[->,bend right=15] (is1) to node[pos=0.75,left,inner sep=4pt] {\scriptsize$0.85$} (p);
      \draw[->,bend left=15] (is1) to node[pos=0.75,left,inner sep=4pt] {\scriptsize$0.15$} (f);

      \draw[->] (i) to[out=-90,in=180] node[pos=0.225,left,align=center,overlay] {$\mathit{stop}$} node[pos=0.5,split] {} node[pos=0.85,below] {\scriptsize$1$} (f);

      \node[split] (ps2) [right=0.9 of p] {};
      \draw[-] (p) to node[below,inner sep=2pt] {$\mathit{subm}$} (ps2);
      \draw[->,shorten >=-1.5pt] (ps2) to[out=-90,in=30] node[pos=0.875,right,align=left,inner sep=5pt,overlay] {\scriptsize$~~\,0.2$\\[-3pt]\scriptsize(\textbf{\texttt{+}4}/\\[-4pt]\scriptsize\textbf{~\texttt{+}24})} (f);
      \draw[->] (ps2) to[out=90,in=70,looseness=1.15] node[pos=0.567,below,inner sep=2.5pt] {\scriptsize$0.8$} node[pos=0.567,above,inner sep=0.5pt] {\scriptsize(0/\textbf{\texttt{+}24})} (p);

      \draw[->] (p) to[bend left=20] node[right,pos=0.33,inner sep=1.5pt,align=left] {$\mathit{arch}$\\[-2pt]\scriptsize\,(\textbf{\texttt{+}1}/0)} node[pos=0.65, split] {} node[right,pos=0.83] {\scriptsize$1$} (f);

      \draw[->] (f) to[loop,out=-60,in=-120,looseness=4] node[pos=0.15, right] {$\tau$}  node[split, pos=0.5, yshift=0.75pt] {} node[left, pos=0.85] {\scriptsize$1$} (f);
    \end{tikzpicture}
    \label{fig:ExampleMDP}
  }
  \subfloat[The true Pareto front]{
    \begin{tikzpicture}[scale=0.6]
      \begin{axis}[
        width=7.1cm,
        height=6.25cm,
        xlabel={maximise $\mathcal{R}_\mathit{rec}$~$\bm\rightarrow$},
        ylabel={$\bm\leftarrow$~minimise $\mathcal{R}_\mathit{eff}$},
        ylabel shift=-6pt,
        legend pos=north west,
        legend cell align={left},
        ymin=0,%
        xmin=0,%
        ymax=120,
        xmax=3.67,
        xtick={0,...,3}
        ]
        \addplot[thick,blue!40!white] (0, 0) -- (0,0);
        \addplot[thick,red!40!white] (0, 0) -- (0,0);
        \addplot[thick,black] (0, 0) -- (0, 0);
        \legend{achievable, unachievable, Pareto front}
        \addplot[blue!40!white, fill=blue!40!white,nearly transparent] (0, 0) -- (0.85, 10) -- (3.4, 112) -- (3.4,120) -- (-1, 120) -- (-1, 0) -- cycle;
        \addplot[red!40!white, fill=red!40!white,nearly transparent] (-1, 0) -- (0, 0) -- (0.85, 10) -- (3.4, 112) -- (3.4, 120) -- (4, 120) -- (4, -10) -- (-1, -10) -- cycle;
        \addplot[thick] (0, 0) -- (0.85, 10) -- (3.4, 112);
        \draw (0, 0) node [above right,align=left,yshift=2pt] {\!\scriptsize(1)};
        \draw (0.85, 10) node [above left,inner sep=1pt,xshift=1pt,yshift=1pt] {\scriptsize(2)};
        \draw (3.4, 112) node [above left,inner sep=1pt,xshift=-1pt,yshift=-2pt] {\scriptsize(3)};
      \end{axis}
    \end{tikzpicture}
    \label{fig:ExampleParetoCurve}
  }
  \subfloat[SMC approximations]{
    \begin{tikzpicture}[scale=0.6]
      \begin{axis}[
        width=7.1cm,
        height=6.25cm,
        xlabel={maximise $\mathcal{R}_\mathit{rec}$~$\bm\rightarrow$},
        ylabel={$\bm\leftarrow$~minimise $\mathcal{R}_\mathit{eff}$},
        ylabel shift=-6pt,
        legend pos=north west,
        legend cell align={left},
        ymin=0,%
        xmin=0,%
        ymax=120,
        xmax=3.67,
        xtick={0,...,3}
        ]

        \addplot[thick,blue!40!white] (0, 0) -- (0,0);
        \addplot[thick,red!40!white] (0, 0) -- (0,0);
        \addplot[thick] (0, 0) -- (0,0);
        \legend{underapprox., ``overapprox.''\phantom{\!\!\!d}, Pareto front}

        \addplot[fill=blue!40!white,nearly transparent] (0, 20) -- (1, 40) -- (2, 80) -- (2.5,110) -- (2.5,120) -- (-1, 120) -- (-1, 20) -- cycle;
        \addplot[thick,blue!80!white,nearly transparent] (-1, 20) -- (0, 20) -- (1, 40) -- (2, 80) -- (2.5,110) -- (2.5,120);
        \addplot[fill=red!40!white,nearly transparent] (-1, 10) -- (0, 10) -- (2.8, 70) -- (3.5, 100) -- (3.5, 120) -- (3.7, 120) -- (4, 120) -- (4, -10) -- (-1, -10) -- cycle;
        \addplot[thick,red!80!white,nearly transparent] (-1, 10) -- (0, 10) -- (2.8, 70) -- (3.5, 100) -- (3.5, 120) -- (3.7, 120);
        \addplot[thick] (0, 0) -- (0.85, 10) -- (3.4, 112);
      \end{axis}
    \end{tikzpicture}
    \label{fig:ExampleCurveApproximation}
  }
  \caption{An example MDP and the corresponding expected-reward Pareto fronts}
  \label{fig:ExampleMDPAndCurves}
\end{figure}

\paragraph{Single-objective properties.}
The optimal \emph{reachability probabilities} of a set of goal states $G \subset \mathcal{S}$ are
$\mathrm{P}_\mathrm{\!min}(\diamond\, G) \defeq \textstyle\inf_{\sigma} \mathrm{P}^\sigma(\diamond\, G)$
and
$\mathrm{P}_\mathrm{\!max}(\diamond\, G) \defeq \textstyle\sup_{\sigma} \mathrm{P}^\sigma(\diamond\, G)$
where
$\mathrm{P}^\sigma(\diamond\, G) \defeq \mathbb{P}_\sigma^{M,s_\mathit{init}}(\Pi_G)$
with $\Pi_G$ the set of paths that contain a state in $G$.
The optimal \emph{expected rewards} to reach a goal are
$\mathrm{E}_\mathrm{min}^\mathcal{R}(\diamond\, G) \defeq \textstyle\inf_{\sigma} \mathrm{E}^{\mathcal{R},\sigma}(\diamond\, G)$
and
$\mathrm{E}_\mathrm{max}^\mathcal{R}(\diamond\, G) \defeq \textstyle\sup_{\sigma} \mathrm{E}^{\mathcal{R},\sigma}(\diamond\, G)$
where
$\mathrm{E}^{\mathcal{R},\sigma}(\diamond\, G) \defeq \mathbb{E}_\sigma^{M,s_\mathit{init}}(R^\mathcal{R}_G)$,
$\mathbb{E}_\sigma^{M,s_\mathit{init}}(r)$ is the expectation of random variable $r$ under $\mathbb{P}_\sigma^{M,s_\mathit{init}}$, and $R^\mathcal{R}_G$ maps each path to the sum of the rewards incurred up to the first state in~$G$, or to $\infty$ if the path does not contain such a state.
Note that we consider undiscounted \emph{reachability rewards}, not total-reward or discounted properties.
We restrict to \emph{contracting} MDPs, \ie MDPs that do not contain nontrivial end components.
This is common in the literature~\cite{HM14,BKLPW17},
suffices for simulation runs to almost surely terminate---%
otherwise, we would need more powerful checks %
that require more than constant memory~\cite{ADKW20}---%
and avoids semantic issues which cause current multi-objective \emph{PMC} tools to return inconsistent results on non-contracting MDPs~\cite{HQW26}.
In this setting, for the properties we consider, DM strategies are optimal, even if the model has both positive and negative rewards~\cite[p.~164]{BKLPW17}.
For a property $\phi = \mathrm{P}_\mathit{\!opt}(\diamond\, G)$ or $\phi = \mathrm{E}_\mathit{opt}^\mathcal{R}(\diamond\, G)$, we write $\phi(\sigma)$ for $\mathrm{P}^\sigma(\diamond\, G)$ or $\mathrm{E}^{\mathcal{R},\sigma}(\diamond\, G)$, respectively.

\paragraph{Multi-objective properties.}
Given  $\phi_1, \ldots, \phi_d$, the \emph{multi-objective property} $\phi = \mathrm{multi}(\phi_1, \ldots, \phi_d)$ asks %
for the Pareto front over all the properties.
The \textit{achievable set} $\Ach{\phi}$ contains all points $\tuple{v_1, \ldots, v_d} \in \RR^d$ for~which
$\exists\,\sigma\colon \phi_1(\sigma) \sim_1 v_1 \wedge \ldots \wedge \phi_d(\sigma) \sim_d v_d$
where ${\sim_i} = {\leq}$ if $\phi_i$ is a ${\cdot}_\mathrm{max}$ property and ${\sim_i} = {\geq}$ else.
A point $v = \tuple{v_1, \ldots, v_d} \in \Ach{\phi}$ is \emph{Pareto-optimal} if no \emph{other} point $\tuple{v_1', \ldots, v_d'}\neq v$ with $\forall i\colon v_i \sim_i v_i'$ is achievable;
these points form the \emph{Pareto~front}.

MDP with multi-objective combinations of reachability probabilities and expected rewards can be transformed into MDP with combinations of only total rewards~\cite[Prop.~2]{FKP12}, for which memoryless strategies are Pareto-optimal~\cite[Sect.~3.4]{Qua23}.\label{text:MemorylessSuffices}
The Pareto front is the convex hull of a finite set of corner points resulting from DM strategies; all its other points are convex combinations of the corners resulting from probabilistic strategies.
Note that, for contracting MDPs, \cite{Qua23} also removes \cite{FKP12}'s restriction to non-negative~rewards.%

An \emph{underapproximation} of the achievable set (and thus the Pareto front) is a convex set of achievable but not necessarily Pareto-optimal points; an \emph{overapproximation}, conversely, is a concave set of unachievable or Pareto-optimal~points.

\begin{example}
\Cref{fig:ExampleParetoCurve} shows the achievable set and Pareto front for $M_R$ and %
$
\phi =
\mathrm{multi}(\mathrm{E}_\mathrm{max}^{\mathcal{R}_\mathit{rec}}(\diamond\, \set{ \mathit{done} }), \mathrm{E}_\mathrm{min}^{\mathcal{R}_\mathit{eff}}(\diamond\, \set{ \mathit{done} }))
$, visualising the tradeoff between maximising recognition and minimising effort. %
The corner points correspond to
(1)~not writing a paper,
(2)~writing a paper and, if written successfully, archiving it on arXiv, and
(3)~(re)submitting it to the conference until it is accepted.
\end{example}

\begin{example}
In \Cref{fig:ExampleCurveApproximation}, the blue (top-left) area is an underapproximation of the achievable set of \Cref{fig:ExampleParetoCurve}.
The red (bottom-left) area is neither an under- nor an overapproximation.
Our methods can statistically approximate these two areas.
\end{example}

\paragraph{Statistical model checking.}
Given a fully stochastic executable formal model (\eg a large Markov chain compactly specified in a high-level modelling language), %
a property $\phi$ that can be evaluated to a real value on a finite path, and user-specified number of runs $n$ and confidence level $\gamma \in (0, 1)$, %
we can perform \emph{statistical model checking} (SMC):
Randomly---using a pseudo-random number generator (PRNG) to implement the model's probabilistic choices---perform $n$ simulation runs, \ie generate $n$ random paths, evaluate $\phi$ on each path to obtain samples $X_1, \ldots, X_n$, compute the sample mean $\hat x = \frac{1}{n}\sum_{i=1}^n X_i$ and a confidence interval (CI) $I = [l, u] \ni \hat x$, and return the pair $\tuple{\hat x, I}$.
Various methods exist to compute \emph{sound} CIs~\cite{BHMWW25a} which guarantee that, when repeating the SMC procedure many times, in the limit $(100 \cdot \gamma)\,\%$ or more of the intervals contain the property's (unknown) true value.
The interval's \emph{precision} is $\varepsilon = \frac{1}{2}(u - l)$; for many methods, we can precompute the $n$ needed to achieve a given~$\varepsilon$.

\paragraph{Lightweight strategy sampling.}
To apply SMC to nondeterministic models like MDPs, we need to simulate under a given strategy.
Many SMC tools implicitly use the uniform random strategy.
\emph{Lightweight strategy sampling} (LSS) allows simulating many different strategies using a constant-memory representation for each, \eg a 32-bit integer.
When simulating with such a \emph{strategy identifier} $\sigma$ and encountering a state $s$ with $|\mathcal{A}(s)| = k > 1$, LSS chooses the $(\mathcal{H}(\sigma . s) \mathbin{\mathrm{mod}} k + 1)$-th action (for some fixed ordering of $\mathcal{A}(s)$) where $\mathcal{H}$ is a hash function and $\sigma . s$ is the concatenation of the bitstring representations of $\sigma$ and~$s$.
In this way, LSS %
implements a DM strategy.
By feeding different information into $\mathcal{H}$, LSS can easily provide history-dependent or partially-informed~\cite{DGHS18,DFH20,BDH24} strategies as well.
Good $\mathcal{H}$ ensure that small changes in the input bits, in particular in $\sigma$, have large effects on the output, to cover as much of the strategy space as uniformly as possible.
Thus strategy identifiers are ``opaque'', \ie their values are not connected to their action choices in any systematic way.\!\footnote{%
Given user input/domain knowledge, biasing to certain choices may increase the probability of sampling near-optimal strategies; in an uninformed black-box setting, however, no fixed bias can improve results for all models, so we sample uniformly.}

An MDP's set of DM strategies---its \emph{strategy space}---is vast but finite.
We use 32-bit strategy identifiers, covering up to $2^{32}$ strategies; while a large space to sample from, it will be smaller than the strategy space of many interesting MDPs.
We consider this a \emph{practical} limitation (that could be mitigated by moving to larger identifiers); in the \emph{theoretical} arguments we make in the following sections, we assume \emph{ideal} LSS where by uniformly sampling identifiers we will almost surely eventually encounter every DM strategy.
We use a PRNG $\mathfrak{P}$ to sample strategy identifiers; by initialising $\mathfrak{P}$ with a fixed \emph{strategy seed}, we can repeatedly sample the same sequence of strategies to make our experiments more repeatable.

\section{Multi-Objective SMC}
\label{sec:MOSMC}

We present an incremental scheme %
and three fixed simulation budget algorithms. %
The latter use heuristics to discard seemingly suboptimal strategies. %
Fix an MDP as in \Cref{def:MDP}, $\phi = \mathrm{multi}(\phi_1, \ldots, \phi_d)$, and confidence level~$\gamma = 1 - \alpha$. %
Recall that the Pareto front's corner points arise from DM strategies after the MDP transformation of \cite[Prop.~2]{FKP12}, which essentially adds one Boolean state variable per property $\phi_i$ to track whether its goal set has been reached.

\paragraph{Approximating the Pareto front.}
All our approaches use LSS to sample deterministic strategies with $d$ bits of memory, which track the visited goal sets to implement the transformation of \cite[Prop.~2]{FKP12} on-the-fly, and evaluate them via SMC. %
For each strategy $\sigma$, we thus obtain $d$ pairs $\tuple{\hat{x}^\sigma_i, I^\sigma_i}$ of a sample mean and a CI with precision $\precision{i}{\sigma}$ for each dimension, $i \in \set{1, \ldots, d}$; together, the one-dimensional CIs describe a $d$-dimensional confidence box $B_\sigma$ around the multi-objective sample mean $\hat{x}^\sigma = \tuple{\hat{x}^\sigma_1, \ldots, \hat{x}^\sigma_d}$.
We ensure that the boxes for the evaluated strategies are \emph{all} simultaneously correct with a-priori probability $\gamma$ (requirement~1).
We write $\mathit{SMC}(\sigma, \gamma, \ldots)$ for the invocation of SMC with LSS for strategy identifier $\sigma$ and statistical parameters $\gamma, \ldots$ that returns $\tuple{\hat{x}_\sigma, B_\sigma}$ satisfying requirement~1.
The convex hull of the ``most pessimistic'' corners of all boxes then bounds a candidate underapproximation $\underline{C}$ of the Pareto front; using
the ``most optimistic'' corners gives a candidate overapproximation~$\overline{C}$.
Given a function $\mathit{stat}$ mapping strategy identifiers to $\tuple{\hat{x}^\sigma, B_\sigma}$ pairs, we also write $\underline{C}(\mathit{stat})$ and $\overline{C}(\mathit{stat})$ to refer to the respective convex-hull computations.

\begin{wrapstuff}[type=figure,width=4.25cm]
  \centering
  \begin{tikzpicture}[scale=0.6]
    \begin{axis}[
      width=7.1cm,
      height=6.25cm,
      xlabel={maximise $\mathcal{R}_\mathit{rec}$~$\bm\rightarrow$},
      ylabel={$\bm\leftarrow$~minimise $\mathcal{R}_\mathit{eff}$},
      ylabel shift=-6pt,
      legend pos=north west,
      legend cell align={left},
      ymin=-10,
      xmin=-0.33,
      ymax=120,
      xmax=3.67,
      xtick={0,...,3}
      ]
      \addplot[thick,blue!40!white] (0, 0) -- (0,0);
      \addplot[thick,red!40!white] (0, 0) -- (0,0);
      \addplot[thick] (0, 0) -- (0,0);
      \legend{underapprox., ``overapprox.''\phantom{\!\!\!d}, Pareto front}

      \addplot[fill=blue!40!white,nearly transparent] (0, 20) -- (1, 40) -- (2, 80) -- (2.5,110) -- (2.5,120) -- (-1, 120) -- (-1, 20) -- cycle;
      \addplot[thick,blue!80!white,nearly transparent] (-1, 20) -- (0, 20) -- (1, 40) -- (2, 80) -- (2.5,110) -- (2.5,120);
      \addplot[fill=red!40!white,nearly transparent] (-1, 10) -- (0, 10) -- (2.8, 70) -- (3.5, 100) -- (3.5, 120) -- (3.7, 120) -- (4, 120) -- (4, -10) -- (-1, -10) -- cycle;
      \addplot[thick,red!80!white,nearly transparent] (-1, 10) -- (0, 10) -- (2.8, 70) -- (3.5, 100) -- (3.5, 120) -- (3.7, 120);
      \addplot[thick] (0, 0) -- (0.85, 10) -- (3.4, 112);

      \filldraw[thick,fill=gray!20!white,nearly transparent,draw=gray!80!white] (0, 10) rectangle (-0.5,40);
      \node[gray] (x1) at (-0.25, 25) [dot] {};

      \filldraw[thick,fill=gray!20!white,nearly transparent,draw=gray!80!white] (0, 20) rectangle (0.2,15);
      \node[gray] (x2) at (0.1, 17.5) [dot] {};

      \filldraw[thick,fill=gray!20!white,nearly transparent,draw=gray!80!white] (3.5, 100) rectangle (2.5,110);
      \node[gray] (x3) at (3, 105) [dot] {};

      \filldraw[thick,fill=gray!20!white,nearly transparent,draw=gray!80!white] (0.3, 40) rectangle (0.8,50);
      \node[gray] (x4) at (0.55, 45) [dot] {};

      \filldraw[thick,fill=gray!20!white,nearly transparent,draw=gray!80!white] (1, 40) rectangle (1.2,37);
      \node[gray] (x5) at (1.1,38.5) [dot] {};

      \filldraw[thick,fill=gray!20!white,nearly transparent,draw=gray!80!white] (2.8, 70) rectangle (1.9,80);
      \node[gray] (x6) at (2.35, 75) [dot] {};

    \end{axis}
  \end{tikzpicture}
  \vspace{-0.3cm}
  \caption{CI boxes and fronts}
  \label{fig:ExampleCurveApproximationConstruction}
\end{wrapstuff}

\begin{example}
\hspace{-4.5ex} %
\label{ex:EstimatedCurves}
Assume $M_R'$ is a larger version of MDP $M_R$ with the same Pareto-optimal strategies but also many other non-optimal DM ones.
\Cref{fig:ExampleCurveApproximationConstruction} plots a possible result of randomly sampling and evaluating six strategies for $M_R'$.
We draw the sample means (grey dots) and confidence boxes for each, and take the top-left corners (because we want to maximise on the horizontal axis and minimise on the vertical axis) to form $\underline{C}$, which is an underapproximation of the true Pareto front.
The most optimistic (bottom-right) corners do not characterise a valid overapproximation because we did not sample good strategies on the bottom left.
\end{example}
\begin{lemma}
\label{thm:UnderapproxConf}
$\underline{C}$ is an underapproximation with probability~$\gamma$.
\end{lemma}

\begin{proof}[sketch]
If a most pessimistic corner is unachievable, then its entire box lies in the unachievable region.
As the true means for all sampled strategies are by definition achievable points, the box cannot contain the true mean.
The probability that this happens for any box is $1 - \gamma$ by requirement~1.
\end{proof}

\begin{remark}
\label{rem:NoOverapprox}
For $\overline{C}$, a similar statement is impossible because we may always have missed a Pareto-optimal strategy.
This is a fundamental limitation of a constant-memory sampling-based approach. %
Only with deeper insights into the model (beyond our assumed black-box view) could we say more, \eg if we knew how likely it is to sample a Pareto-optimal strategy---but finding that probability essentially solves the multi-objective model checking problem already.
\end{remark}

\paragraph{Multiple comparisons.}
We perform SMC (1)~for each strategy (2)~for each dimension, but want all confidence \emph{boxes} to be correct \emph{simultaneously} with probability~$\gamma$.
If we build a CI for each strategy and dimension with confidence~$\gamma$, then \emph{at least one} CI is wrong with probability~$\gg 1-\gamma$.
There are two common ways to counteract this \emph{multiple comparisons problem}:
(1)~In case of statistical independence, \eg for the different strategies (since we perform separate simulation runs for each), Šidák correction~\cite{Sid67} can be applied.
That is, when sampling $m$ strategies, request confidence $\gamma_\sigma = \gamma^\frac{1}{m}$ for each strategy.
(2)~Within a call to $\mathit{SMC}(\sigma, \gamma_\sigma, \ldots)$, all dimensions are estimated from the same simulation runs, so their CIs are not statistically independent.
Then Bonferroni correction~\cite{Mil66} resp.\ the union bound can be used, which requires $\sum_{i=1}^k \alpha_i \leq \alpha$ for $k$ CIs.
Thus the error budget for each dimension can be set to \eg $\alpha_\sigma/d$ with $\alpha_\sigma = 1 - \gamma_\sigma$. %
Šidák correction is harder to work with if we want to distribute the error budget unevenly, and only slightly more efficient than Bonferroni/the union bound beyond very small numbers of comparisons~\cite{Abd07}.
We thus use the latter throughout.

\subsection{Incremental Sampling with Convergence}
\label{sec:incremental}

\begin{wrapstuff}[type=figure,width=0.51\textwidth]
\hspace{-1.5mm}
\begin{minipage}[t]{6.75cm}
\vspace{-3.25mm}
\begin{algorithm}[H]
    \Function{$\mathit{IncSamp}(\text{PRNG } \mathfrak{P}, m, \alpha, \varepsilon, f)$}{
        $\mathit{stat} := \emptyset$, $\underline{C} := \emptyset$, $\overline{C} := \emptyset$\;
        \While{not interrupted}{
            \For{$i=1$ to $m$}{
                $\sigma$ := next sample from $\mathfrak{P}$\;
                $\mathit{stat}(\sigma) := \mathit{SMC}(\sigma, 1 - \frac{f \cdot \alpha}{m}, \varepsilon)$\label{alg:incremental:SMC}\;
                $\underline{C} := \underline{C}(\mathit{stat})$, $\overline{C} := \overline{C}(\mathit{stat})$\;
            }
            $\alpha := (1-f) \cdot \alpha$\;
        }
        \Return $\tuple{ \underline{C}, \overline{C} }$\;
    }
    \caption{Incremental Sampling}\label{alg:incremental}
\end{algorithm}
\end{minipage}
\end{wrapstuff}

\noindent
Our incremental scheme indefinitely samples batches of $m$ strategies as shown in \Cref{alg:incremental}, updating $\underline{C}$ and $\overline{C}$ whenever a strategy has been evaluated.
When interrupted, %
it returns the current candidate approximations $\underline{C}$ and $\overline{C}$.
For \Cref{thm:UnderapproxConf} to hold at any interruption point, we need to divide the error budget $\alpha$ over an a priori unknown number of batches.
We use the union bound and a user-specified factor~$f \in (0,1)$:
The $i$-th batch gets error budget $(1 - f)^{i-1} \cdot f \cdot \alpha$.
Each batch's budget is then divided equally among the $m$ strategies.
We parametrise the SMC invocations in line~\ref{alg:incremental:SMC} with precision $\varepsilon$, not the number of runs $n$; in this way, as $\alpha$ decreases from batch to batch, the number of runs performed per strategy increases.
If we fixed $n$ instead, then $\varepsilon$ would increase, resulting in larger and larger boxes for later strategies.

\begin{lemma}
\label{thm:IncrementalApprox}
With probability $\gamma = 1 - \alpha$:
\begin{enumerate}[label={(\arabic*)}]
\item
When interrupted, $\underline{C}$ returned by \Cref{alg:incremental} is an underapproximation.
\item
When not interrupted, $\underline{C}$ and $\overline{C}$ inside \Cref{alg:incremental} almost surely converge to an under- and overapproximation with precision $2\sqrt{d\varepsilon^2}$, respectively.
\end{enumerate}
\end{lemma}

\begin{proof}[sketch]
(1)~follows directly from \Cref{thm:UnderapproxConf} and \Cref{alg:incremental}'s construction (\ie the use of the union bound to distribute~$\alpha$).
For (2), observe that, when sampling indefinitely with ideal LSS on finite MDPs, almost surely all Pareto-optimal strategies $\sigma^\mathit{opt}_i$ will be sampled eventually.
At that (unknown) point, with probability $\gamma$ (by \Cref{alg:incremental}'s construction), all of the $B_{\sigma^\mathit{opt}_i}$ contain the strategy's true mean, and thus $\overline{C}$ is an overapproximation.
In the worst case, the true mean is at the most pessimistic corner, so at a ($d$-dimensional Euclidean) distance of $2\sqrt{d\varepsilon^2}$ from the most optimistic corner, or vice-versa.
\end{proof}
Thus, in the limit, we are $\gamma$-confident that the true Pareto front is fully enveloped by $\underline{C}$ and $\overline{C}$:
they form a \emph{simultaneous $\gamma$ confidence band} for the true front.

\begin{wrapstuff}[type=figure,width=0.45\textwidth]
\begin{tikzpicture}[scale=0.6]
    \begin{axis}[
        name=Axis,
        xlabel={Sampled runs},
        ylabel={Sampled strategies},
        legend pos=north west,
        legend cell align=left,
        ymin=0,
        xmin=0,
        xmax=1000000000,
        xticklabel style={/pgf/number format/fixed}
    ]
        \addplot[thick,blue,dotted,mark=none,restrict x to domain=0:1000000000] table[x=a1,y=a2] {data/approximation-settings.dat};
        \addplot[thick,red,dashed,mark=none,restrict x to domain=0:1000000000] table[x=b1,y=b2] {data/approximation-settings.dat};
        \addplot[thick,orange,loosely dashdotted,mark=none,restrict x to domain=0:1000000000] table[x=c1,y=c2] {data/approximation-settings.dat};
        \addplot[mark=none,restrict x to domain=0:1000000000] table[x=d1,y=d2] {data/approximation-settings.dat};
        \legend{$f=0.1;m=100$, $f=0.1;m=1000$, $f=0.5;m=100$, $f=0.5;m=1000$}
    \end{axis}
\end{tikzpicture}
\vspace{-0.2cm}
\caption{Sampled strategies vs. runs}
\label{fig:approximation-settings:middle}
\end{wrapstuff}

\Cref{thm:IncrementalApprox} holds independently of the values chosen for $m$ and $f$.
The error budget assigned to a strategy determines the number of simulation runs needed in $\mathit{SMC}(\sigma, \gamma', \varepsilon)$.
A larger $f$ results in a higher budget for earlier batches, requiring more runs in later batches to achieve the fixed precision~$\varepsilon$.
A larger batch size $m$ distributes each batch's budget over more strategies, lowering the budget for earlier strategies. %
An example of these tradeoffs is visualised in \Cref{fig:approximation-settings:middle}, using the Clopper-Pearson interval~\cite{CP34} and $d = 2, \varepsilon = 0.01, \alpha = 0.1$.
On the horizontal axis, the number of runs performed is a proxy for the runtime needed; completing the evaluation of more strategies (higher values on the vertical axis) with fewer runs is ``better''.
As expected, the combination of high $f = 0.5$ and low $m = 100$ starts fast but is overtaken by more conservative approaches in the long run.
For large batch sizes, however, the influence of dividing $\alpha$ is small in practice.

\subsection{Fixed-Budget Sampling}
\label{sec:fixedbudget}

In practice, %
we only have finite time.
We propose three %
algorithms---\emph{weight vector refinement} (WVR), \emph{fixed iteration budget} (FIB), and \emph{fixed strategy budget} (FSB)---that try to statistically approximate the achievable set as closely as possible in a fixed simulation budget.
They use heuristics to select ``promising'' strategies, performing fewer simulation runs for seemingly suboptimal ones so that promising ones can be estimated more precisely or more strategies can be sampled. %
We describe our heuristics in \Cref{sec:heuristics}; for now, assume one is given as a function $\selcandidates\colon \powerset{\Sigma} \to \powerset{\Sigma}$ that, given a candidate set of strategies, returns a subset of promising strategies.
WVR adapts multi-objective PMC~\cite{FKP12,Qua23} to the SMC setting, while FIB transfers smart sampling~\cite{DLST15} from single-objective to multi-objective SMC.
FSB adds dynamic addition of strategies to FIB.
While many more schemes are imaginable, these are fundamental corner points.

Since only finitely many strategies can be sampled in finite time, %
any $\overline{C}$ we compute may not be an overapproximation (as in \Cref{ex:EstimatedCurves}); and by \Cref{rem:NoOverapprox}, we cannot quantify the chance for this happening either.
The fixed-budget algorithms thus focus on $\underline{C}$, which they guarantee to be an underapproximation with probability $\gamma$ again.
They consist of a heuristic phase that determines a set of candidate strategies, followed by an evaluation phase that returns~$\underline{C}$. %
The latter samples new runs to ensure statistical independence between the phases to avoid bias %
(see\label{ref:ExplanationAvoidBias} \Cref{ap:Explanations} for an intuitive explanation).
\Cref{alg:evaluate}, the evaluation phase, is the same for WVR, FIB, and FSB; the heuristic phases differ, but all perform $I \cdot m \cdot n$ simulation runs (up to rounding) for $I$ a number of internal iterations.
Our implementation then uses $I \cdot m \cdot n$ more runs for the evaluation~phase.

\begin{algorithm}[t]
    \Function{$\mathit{Eval}(\text{strategy set } \Sigma, \text{run budget}\; n, \text{statistical error budget } \alpha)$}{
        $\mathit{stat} := \emptyset$\; %
        \lForEach{$\sigma \in \Sigma$}{%
            $\mathit{stat}(\sigma) := \mathit{SMC}(\sigma, 1 - \frac{\alpha}{|\Sigma|}, \lfloor\frac{n}{|\Sigma|}\rfloor)$%
        }
        \Return $\underline{C}(\mathit{stat})$\;
    }
    \caption{Fixed-budget SMC for multiple objectives, evaluation phase.}\label{alg:evaluate}
\end{algorithm}

\begin{algorithm}[t]
    \Function{$\mathit{WVR}(\text{PRNG } \mathfrak{P}, \text{heuristic}\; \rho, m, n, I, \alpha)$}{
        $\Sigma := \text{next } \lfloor\frac{m}{2}\rfloor \text{ samples from }\mathfrak{P}$, $\mathit{stat} := \emptyset$\;
        \lForEach{$\sigma \in \Sigma$}{%
            $\mathit{stat}(\sigma) := \mathit{SMC}(\sigma, 1 - \frac{\alpha}{m}, n)$
        }
        \For{$i=2$ to $I$}{
            select $w \in \mathbb{R}^d$ in direction of max.\ distance between $\underline{C}(\mathit{stat}), \overline{C}(\mathit{stat})$\label{alg:weightvector:SelectW}\;
            $\Sigma' := \text{select } \lfloor\frac{|\Sigma|}{2}\rfloor \text{ best strategies in }\Sigma\text{ according to } \lambda\,\sigma \text{. } w\cdot \est{}{\sigma}$\;
            \DoWhile(\label{alg:weightvector:SSStart}){$|\Sigma'| > 0$\label{alg:weightvector:SSEnd}}{
                \lForEach{$\sigma \in \Sigma'$}{%
                    $\mathit{stat}(\sigma) := \mathit{SMC}(\sigma, 1 - \frac{\alpha}{m}, n)$
                }
                $\Sigma' := \text{select } \lfloor\frac{|\Sigma'|}{2}\rfloor \text{ best strategies in }\Sigma\text{ according to } \lambda\,\sigma \text{. } w\cdot \est{}{\sigma}$\;
            }
        }

        \Return $\mathit{Eval}(\selcandidates_\rho(\Sigma), I \cdot m \cdot n, \alpha)$\;
    }
    \caption{The weight vector refinement (WVR) algorithm.}\label{alg:weightvector}
\end{algorithm}

\paragraph{Weight vector refinement.}
Listed in \Cref{alg:weightvector}, our WVR algorithm implements the heuristic phase by \emph{linear scalarisation}:
it uses a specified number $I$ of weight vectors to convert the multi-objective problem into $I$ single-objective problems, as is done in multi-objective PMC, RL~\cite{YNIT19,SXWL+22,MDN13}, and optimisation~\cite{QAVO19,Sin87,WHQLL17}.
When determining the next weight vector, we pick the one with the largest distance between $\underline{C}$ and $\overline{C}$ (line~\ref{alg:weightvector:SelectW}) as in~\cite{Qua23}.
For each vector, we apply a round of smart strategy sampling~\cite{DLST15} in lines \ref{alg:weightvector:SSStart}-\ref{alg:weightvector:SSEnd}, with one exception:
The worst-performing strategies are discarded from the local $\Sigma'$ only, %
so that they will (temporarily) not receive more runs under the current weight vector.
Every \textbf{do} loop iteration performs more runs for the remaining strategies to estimate the most promising ones (those with the highest value in the weight vector's direction) most~precisely.

\begin{algorithm}[t]
    \Function{$\mathit{FIB}(\text{PRNG } \mathfrak{P}, \text{heuristic}\; \rho, m, n, I, \alpha)$}{
        $\Sigma_1 := \text{next } m \text{ samples from }\mathfrak{P}$, $\mathit{stat} := \emptyset$\;
        \For{$i=1$ to $I$}{
            \lForEach{$\sigma \in \Sigma_i$}{%
                $\mathit{stat}(\sigma) := \mathit{SMC}(\sigma, 1 - \frac{\alpha}{m}, \lfloor\frac{n\cdot m}{|\Sigma_i|}\rfloor)$\label{alg:fixediterationbudget:PerformSMC}
            }
            $\Sigma_{i+1} = \selcandidates_\rho(\Sigma_i)$\label{alg:fixediterationbudget:Discard}\;
        }
        \Return $\mathit{Eval}(\selcandidates_\rho(\Sigma_{I+1}), I \cdot m \cdot n, \alpha)$\;
    }
    \caption{The fixed iteration budget (FIB) algorithm.}\label{alg:fixediterationbudget}
\end{algorithm}

Linear scalarisation is a necessity in PMC, to turn $\phi$ into single-objective properties for which PMC algorithms are available.
SMC, however, can estimate the values of multiple objectives at once, enabling more direct approaches:

\paragraph{Fixed iteration budget.}
The FIB algorithm shown in \Cref{alg:fixediterationbudget} samples an initial set of strategies $\Sigma_1$, then $I$ times
performs simulation runs for the current strategy set $\Sigma_i$ (line~\ref{alg:fixediterationbudget:PerformSMC}) and
heuristically discards unpromising strategies to form the next set $\Sigma_{i+1}$ (line~\ref{alg:fixediterationbudget:Discard}).
Each iteration evenly distributes $m \cdot n$ runs over the strategies in $\Sigma_i$, so that promising strategies are estimated ever more precisely. %

\begin{algorithm}[t]
    \Function{$\mathit{FSB}(\text{PRNG } \mathfrak{P}, \text{heuristic}\; \rho, m, n, I, \alpha)$}{
        $\Sigma_1 := \varnothing$, $\Sigma'_1 := \text{next } m \text{ samples from }\mathfrak{P}$, $\mathit{stat} := \emptyset$\;
        \For{$i=1$ to $I$}{
            \lForEach{$\sigma \in \Sigma_i$}{%
                $\mathit{stat}(\sigma) := \mathit{SMC}(\sigma, 1 - \frac{\alpha}{m}, n)$
            }
            \lForEach{$\sigma \in \Sigma'_i$}{%
                $\mathit{stat}(\sigma) := \mathit{SMC}(\sigma, 1 - \frac{\alpha}{m}, i \cdot n)$
            }
            $\Sigma_{i+1} := \selcandidates_\rho(\Sigma_i\cup\Sigma'_i)$\;
            $\Sigma'_{i+1} := \text{next } \lfloor\frac{m-|\Sigma_{i+1}|}{i+1}\rfloor \text{ samples from }\mathfrak{P}$\;
        }

        \Return $\mathit{Eval}(\selcandidates_\rho(\Sigma_{I+1}), I \cdot m \cdot n, \alpha)$\;
    }
    \caption{The fixed strategy budget (FSB) algorithm.}\label{alg:fixedstrategybudget}
\end{algorithm}

\paragraph{Fixed strategy budget.}
Instead of redistributing the budget for discarded strategies over the remaining strategies, FSB (shown in \Cref{alg:fixedstrategybudget}) samples new strategies so that, in the end, $I \cdot n$ runs were performed for each ultimately surviving strategy.
For multi-objective properties with many suboptimal strategies, this increases %
the probability of finding a Pareto-optimal strategy.

\subsection{Strategy Selection Heuristics}
\label{sec:heuristics}

\begin{wrapstuff}[type=figure,width=0.33\textwidth,lines=10]
\centering
\begin{tikzpicture}[scale=.33,every node/.style={font=\scriptsize}]
\draw[-] (0,0) -- (10.5,0);
\draw[-] (0,0) -- (0,10.5);
\node at (5.1,-.5) {maximise $x_1 \bm{\rightarrow}$};
\node at (-.5,5.1) {\rotatebox{90}{$\bm{\leftarrow} \text{minimise } x_2$}};

\filldraw[fill=blue!20!white,nearly transparent,draw=blue!80!white] (4.5,1.7) rectangle +(3.8,3.2) ;
\filldraw[fill=red!20!white,nearly transparent,draw=red!80!white] (3.8,2.7) rectangle +(3,4) ;

\node (vf0) at (0,.5) [null] {};
\node (nf0) at (1.5,.5) [dot] {};
\node (vf1) at (1.5,1.2) [null] {};
\node (nf1) at (1.8,1.2) [dot] {};
\node (vf2) at (1.8,2) [null] {};
\node (nf2) at (3.1,2) [dot] {};
\node (vf3) at (3.1,3.3) [null] {};
\node[blue] (nf3) at (6.4,3.3) [dot] {};
\node (vf4) at (6.4,7) [null] {};
\node (nf4) at (8.2,7) [dot] {};
\node (vf5) at (8.2,8.5) [null] {};
\node (nf5) at (9.5,8.5) [dot] {};
\node (vf6) at (9.5,10.2) [null] {};

\node[red] (exc) at (5.3,4.7) [dot] {};

\node at (7.8,2.6) {\textcolor{blue}{$(\est{1}{\varsigma},\est{2}{\varsigma})$}};
\node at (3.9,5.4) {\textcolor{red}{$(\est{1}{\sigma},\est{2}{\sigma})$}};

\draw[red] (5.3,7.5) -- node[above] {$\precision{1}{\sigma}$} (6.8,7.5);
\draw[red] (5.3,7.3) -- (5.3,7.7);
\draw[red] (6.8,7.3) -- (6.8,7.7);
\draw[red] (2.8,4.7) -- node[left] {$\precision{2}{\sigma}$} (2.8,2.7);
\draw[red] (2.6,4.7) -- (3,4.7);
\draw[red] (2.6,2.7) -- (3,2.7);

\draw[blue] (6.4,1.2) -- node[below] {$\precision{1}{\varsigma}$} (4.5,1.2);
\draw[blue] (6.4,1) -- (6.4,1.4);
\draw[blue] (4.5,1) -- (4.5,1.4);
\draw[blue] (9,3.3) -- node[right] {$\precision{2}{\varsigma}$} (9,4.9);
\draw[blue] (8.8,3.3) -- (9.2,3.3);
\draw[blue] (8.8,4.9) -- (9.2,4.9);

\path[-]
(vf0) edge[] node [] {} (nf0)
(nf0) edge[] node [] {} (vf1)
(vf1) edge[] node [] {} (nf1)
(nf1) edge[] node [] {} (vf2)
(vf2) edge[] node [] {} (nf2)
(nf2) edge[] node [] {} (vf3)
(vf3) edge[] node [] {} (nf3)
(nf3) edge[] node [] {} (vf4)
(vf4) edge[] node [] {} (nf4)
(nf4) edge[] node [] {} (vf5)
(vf5) edge[] node [] {} (nf5)
(nf5) edge[] node [] {} (vf6)
;
\end{tikzpicture}
\vspace{-0.3cm}
\caption{The simple rule}
\label{fig:rule:pareto-simple}
\end{wrapstuff}

\noindent
We now describe our strategy selection heuristics, beginning with a simple rule:

\medskip\noindent
\refstepcounter{definition}\label{def:rule-simple}%
\textbf{Definition~\thedefinition.}\enspace
\emph{$\sigma\notin\selcandidates_\mathit{simple}(\Sigma)$ iff there is a witness strategy $\varsigma\in\Sigma$ such that, for all $i \in \set{1, \ldots, d}$,
$\est{i}{\sigma} \leq \est{i}{\varsigma}$
if the $i$-th objective is maximising, and
$\est{i}{\sigma} \geq \est{i}{\varsigma}$
if it is minimising.}

\medskip\noindent
This rule effectively constructs a (non-convex) Pa\-re\-to front from the sample means of the deterministic strategies in $\Sigma$ as illustrated in \Cref{fig:rule:pareto-simple}, discarding all strategies with suboptimal sample means: %
In the figure, %
$\sigma$ will be discarded because witness $\varsigma$ has a better sample mean.

\begin{wrapstuff}[type=figure,l,width=0.475\textwidth,lines=24,top=8]
    \centering
    \scriptsize
    \subfloat[Simple rule]{
        \begin{tikzpicture}[scale=.275]
            \draw[-] (0,0) -- (8.3,0);
            \draw[-] (0,0) -- (0,7.7);
            \node at (4.5,-.5) {$\text{maximise } x_1 \bm{\rightarrow}$};
            \node at (-.5,4) {\rotatebox{90}{$\bm{\leftarrow} \text{minimise } x_2$}};
            \filldraw[fill=blue!20!white,nearly transparent,draw=blue!80!white] (2.8,1) rectangle +(5.2,4);
            \filldraw[fill=red!20!white,nearly transparent,draw=red!80!white] (1.7,2.1) rectangle +(3.2,5);

            \node[blue] (node) at (5.4,3) [dot] {};
            \node[red] (excl) at (3.3,4.6) [dot] {};

            \node at (5.9,2.5) {\textcolor{blue}{$\est{}{\varsigma}$}};
            \node at (2.9,5) {\textcolor{red}{$\est{}{\sigma}$}};

            \draw[<->,red] (excl) -- (node);
        \end{tikzpicture}
        \label{fig:all-rules:simple}
    }
    \subfloat[Far enough]{
        \begin{tikzpicture}[scale=.275]
            \draw[-] (0,0) -- (8.3,0);
            \draw[-] (0,0) -- (0,7.7);
            \node at (4.5,-.5) {$\text{maximise } x_1 \bm{\rightarrow}$};
            \node at (-.5,4) {\rotatebox{90}{$\bm{\leftarrow} \text{minimise } x_2$}};

            \filldraw[fill=blue!20!white,nearly transparent,draw=blue!80!white] (4.6,.5) rectangle +(2.8,5.6);
            \filldraw[fill=red!20!white,nearly transparent,draw=red!80!white] (.6,4.3) rectangle +(6,2.4);

            \node[blue] (node) at (6,3) [dot] {};
            \node[red] (excl) at (3.6,5.5) [dot] {};

            \node at (6.5,2.5) {\textcolor{blue}{$\est{}{\varsigma}$}};
            \node at (3.4,5.7) {\textcolor{red}{$\est{}{\sigma}$}};

            \draw[<->,red] (excl) -- (4.6,5.5) ;
            \draw[<->,red] (6,4.3) -- (node);
        \end{tikzpicture}
        \label{fig:all-rules:far-enough}
    }\\
    \subfloat[Far from excluded]{
        \begin{tikzpicture}[scale=.275]
            \draw[-] (0,0) -- (8.3,0);
            \draw[-] (0,0) -- (0,7.7);
            \node at (4.5,-.5) {$\text{maximise } x_1 \bm{\rightarrow}$};
            \node at (-.5,4) {\rotatebox{90}{$\bm{\leftarrow} \text{minimise } x_2$}};

            \filldraw[fill=blue!20!white,nearly transparent,draw=blue!80!white] (2.2,.4) rectangle +(5.4,5.8);
            \filldraw[fill=red!20!white,nearly transparent,draw=red!80!white] (1.3,4.4) rectangle +(2.6,2.8);

            \node[blue] (node) at (5,3.3) [dot] {};
            \node[red] (excl) at (2.6,5.8) [dot] {};

            \node at (5.5,2.8) {\textcolor{blue}{$\est{}{\varsigma}$}};
            \node at (2.2,6.2) {\textcolor{red}{$\est{}{\sigma}$}};

            \draw[<->,red] (3.9,4.4) -- (node);
        \end{tikzpicture}
        \label{fig:all-rules:far-from-excluded}
    }
    \subfloat[Far from witness]{
        \begin{tikzpicture}[scale=.275]
            \draw[-] (0,0) -- (8.3,0);
            \draw[-] (0,0) -- (0,7.7);
            \node at (4.5,-.5) {$\text{maximise } x_1 \bm{\rightarrow}$};
            \node at (-.5,4) {\rotatebox{90}{$\bm{\leftarrow} \text{minimise } x_2$}};

            \filldraw[fill=blue!20!white,nearly transparent,draw=blue!80!white] (4.5,0.9) rectangle +(2.6,2.4);
            \filldraw[fill=red!20!white,nearly transparent,draw=red!80!white] (.9,1.7) rectangle +(5.4,5.6);

            \node[blue] (node) at (5.8,2.1) [dot] {};
            \node[red] (excl) at (3.6,4.5) [dot] {};

            \node at (6.3,1.3) {\textcolor{blue}{$\est{}{\varsigma}$}};
            \node at (3.4,4.7) {\textcolor{red}{$\est{}{\sigma}$}};

            \draw[<->,red] (excl) -- (4.5,3.3);
        \end{tikzpicture}
        \label{fig:all-rules:far-from-witness}
    }\\
    \subfloat[Far from each other]{
        \begin{tikzpicture}[scale=.275]
            \draw[-] (0,0) -- (8.3,0);
            \draw[-] (0,0) -- (0,7.7);
            \node at (4.5,-.5) {$\text{maximise } x_1 \bm{\rightarrow}$};
            \node at (-.5,4) {\rotatebox{90}{$\bm{\leftarrow} \text{minimise } x_2$}};

            \filldraw[fill=blue!20!white,nearly transparent,draw=blue!80!white] (4.5,0.6) rectangle +(3,3.8);
            \filldraw[fill=red!20!white,nearly transparent,draw=red!80!white] (1.8,3.5) rectangle +(3.2,3.8);

            \node[blue] (node) at (6,2.5) [dot] {};
            \node[red] (excl) at (3.4,5.5) [dot] {};

            \node at (6.5,2) {\textcolor{blue}{$\est{}{\varsigma}$}};
            \node at (3.2,5.7) {\textcolor{red}{$\est{}{\sigma}$}};

            \draw[<->,red] (excl) -- (4.5,4.4);
            \draw[<->,red] (5,3.5) -- (node);
        \end{tikzpicture}
        \label{fig:all-rules:far-from-each-other}
    }
    \subfloat[Conservatively far]{
        \begin{tikzpicture}[scale=.275]
            \draw[-] (0,0) -- (8.3,0);
            \draw[-] (0,0) -- (0,7.7);
            \node at (4.5,-.5) {$\text{maximise } x_1 \bm{\rightarrow}$};
            \node at (-.5,4) {\rotatebox{90}{$\bm{\leftarrow} \text{minimise } x_2$}};

            \filldraw[fill=blue!20!white,nearly transparent,draw=blue!80!white] (5.2,0.8) rectangle +(2.6,2.2);
            \filldraw[fill=red!20!white,nearly transparent,draw=red!80!white] (1.3,4.3) rectangle +(2.6,2.4);

            \node[blue] (node) at (6.5,1.9) [dot] {};
            \node[red] (excl) at (2.6,5.5) [dot] {};

            \node at (7,1.4) {\textcolor{blue}{$\est{}{\varsigma}$}};
            \node at (2.4,5.7) {\textcolor{red}{$\est{}{\sigma}$}};

            \draw[<->,red] (5.2,3) -- (3.9,4.3);
        \end{tikzpicture}
        \label{fig:all-rules:conservatively-far}
    }
    \caption{Strategy discarding rules.}
    \vspace{-0.2cm}
    \label{fig:all-rules}
\end{wrapstuff}

The simple rule only compares sample means, not taking any statistical considerations into account. %
In \Cref{fig:rule:pareto-simple}, the confidence boxes for $\sigma$ and $\varsigma$ overlap, so with the confidence level used, we cannot exclude that $\sigma$ is in fact the better strategy.
A more ``precise'' rule may decide to perform a probabilistic analysis and reject $\sigma$ if it is ``far enough'' from the witness $\varsigma$, where ``far enough'' means that the inequalities of \Cref{def:rule-simple} hold with sufficient probability. %
To simplify the process, we only consider CIs to define the criteria for rejection, not the full sample distributions, using the precisions of the CIs to establish ``safe'' distances.

\Cref{fig:all-rules} sketches all criteria we consider based on the sample means (dots) and confidence boxes for $\sigma$ and witness~$\varsigma$;
the double-ended arrows indicate the exclusion rules.
Following the simple rule, all situations in \Cref{fig:all-rules} exclude $\sigma$ as a consequence of the presence of $\varsigma$.
The more precise rules are defined as follows for maximising objectives; the criterion when minimising is the same with $+$ and $-$ as well as $\leq$ and $\geq$ switched:

\begin{definition}
\hspace{-4.5ex} %
$\selcandidates_\mathit{rule}(\Sigma)$ does \underline{not} contain $\sigma$ iff there is a witness strategy $\varsigma\in\Sigma$ s.t. $\forall i \in \set{1, \ldots, d}$,
$$
\begin{aligned}
\est{i}{\sigma} &\leq \est{i}{\varsigma}-\min(\precision{i}{\varsigma},\precision{i}{\sigma})
&&\quad \textit{(FE)}\\
\est{i}{\sigma}+\precision{i}{\sigma} &\leq \est{i}{\varsigma}
&&\quad \textit{(FFE)}\\
\est{i}{\sigma} &\leq \est{i}{\varsigma}-\precision{i}{\varsigma}
&&\quad \textit{(FFW)}\\
\est{i}{\sigma} &\leq \est{i}{\varsigma}-\max(\precision{i}{\varsigma},\precision{i}{\sigma})
&&\quad \textit{(FFEO)}\\
\est{i}{\sigma}+\precision{i}{\sigma} &\leq \est{i}{\varsigma}-\precision{i}{\varsigma}
&&\quad \textit{(CF)}
\end{aligned}
$$
\end{definition}
\noindent
The rules involving $\varepsilon$ assume symmetric CIs, but extend analogously to asymmetric ones.
Rule ``far enough'' (FE, \Cref{fig:all-rules:far-enough}) ensures that, for every objective, at least one of the values estimated with one strategy is outside the other strategy's CIs.
Strategy $\varsigma$ cannot exclude $\sigma$ in \Cref{fig:all-rules:simple} by this rule since $\est{2}{\sigma}<\est{2}{\varsigma}+\min(\precision{2}{\varsigma},\precision{2}{\sigma})$, but it can do so in \Cref{fig:all-rules:far-enough} since $\est{1}{\sigma}\leq\est{1}{\varsigma}-\precision{1}{\varsigma}$ and $\est{2}{\sigma}\geq\est{2}{\varsigma}+\precision{2}{\sigma}$, and in all of \Cref{fig:all-rules:far-from-excluded,fig:all-rules:far-from-witness,fig:all-rules:far-from-each-other,fig:all-rules:conservatively-far}.

The next rules focus on the confidence of \emph{one} of the strategies: %
Rule ``far from excluded'' (FFE) requires the estimated values for witness $\varsigma$ to be outside the CIs for the excluded strategy $\sigma$,
shown in \Cref{fig:all-rules:far-from-excluded}.
Dual to this is ``far from witness'' (FFW), shown in \Cref{fig:all-rules:far-from-witness}.
These two rules are incomparable but implied by ``simple'' and ``far enough''.
The intersection of both rules is ``far from each other'' (FFEO), visualised in \Cref{fig:all-rules:far-from-each-other}. %
This rule also excludes $\sigma$ in the situation of \Cref{fig:all-rules:conservatively-far}, but will not exclude it in the situations of \Cref{fig:all-rules:simple,fig:all-rules:far-enough,fig:all-rules:far-from-excluded,fig:all-rules:far-from-witness}.
The final and most conservative rule is ``conservatively far'' (CF), which ensures that the CIs obtained with the excluded strategy $\sigma$ do not overlap with the CIs obtained with the witness $\varsigma$ as shown in \Cref{fig:all-rules:conservatively-far}.

\section{Experimental Evaluation}
\label{sec:Experiments}

We implemented the algorithms and heuristics of \Cref{sec:MOSMC} in \modes to experimentally study their effectiveness and performance. %
Our implementation follows the given pseudocode, except that in the fixed-budget algorithms, any budget unused due to rounding effects in one iteration is carried over to the next.

\subsection{Benchmark Set}

We compiled a representative benchmark set consisting of benchmarks from PMC and RL.
For the former, we use all multi-objective models from \cite[Sect.~6.2]{ABBC+24} that \modes supports (in terms of modelling features) and that have nontrivial Pareto fronts (\ie of more than one Pareto-optimal point, as computed via PMC).
These models are small enough to be checked with PMC.
To demonstrate that our approaches keep working when PMC starts to struggle with state space explosion and higher dimensions, we also consider models commonly used in RL literature~\cite{FANB+23,BBS95,Sut95,LTZCX24}:
(1)~the breakable bottles model~\cite{VFDB21} with 3 objectives,
(2)~the fruit tree model~\cite{YSN19} with both 2 and 6 objectives,
(3)~deterministic and probabilistic versions of the deep sea treasure case~\cite{VDBID11},~and
(4)~racetrack benchmarks~\cite{Gar73}, which we make multi-objective via the puddle world~\cite{BM94} reward scheme.
For details, see \Cref{ap:Models}.
Altogether, we end up with 15 models.
Most are parametrised; we consider up to three different parameter valuations per model and one multi-objective property, giving us a final set of 40 benchmark instances.
The benchmark selection was done prior to running any experiments, and was not based on any expectations of the performance of our approaches.

\subsection{Experimental Setup}

To be able to assess the impact of the number of initial strategies $m$ and the run count parameter $n$ on the fixed-budget algorithms, we use three different
\begin{wraptable}[6]{r}{3.7cm}
    \vspace{-1.04em}\centering
    \caption{Configurations}
    \vspace{4pt}
    \label{tab:config}
    \scriptsize
    \setlength{\tabcolsep}{2.5pt}
    \centering
    \begin{tabular}{cccccccccc}
        \toprule
        Conf & $m$ & $\alpha$ & $n$ & $I$ \\
        \midrule
        1 & $100$ & $0.1$ & $1000$ & $10$ \\
        2 & $1000$ & $0.1$ & $100$ & $10$ \\
        3 & $3333$ & $0.1$ & $30$ & $10$ \\
        \bottomrule
    \end{tabular}
    \vspace{-3em}
\end{wraptable}
configurations as listed in \Cref{tab:config}, and
evaluate each algorithm on each instance with all heuristics in all configurations, leading to $3 \cdot 40 \cdot 6 \cdot 3 = 2160$ different experiments.
WVR samples $\frac{m}{2}$, FIB $m$, and FSB at most $\approx m \cdot \ln I$ strategies.
We thus sample $<10000$, so \emph{much} fewer than $2^{32}$, strategies in each experiment.

We judge each setting by its result's \emph{hypervolume} (HV)~\cite{ZT98}:
the area under the Pareto front with respect to a ``worst possible'' reference point.
In \Cref{fig:ExampleCurveApproximation}, a reasonable reference point is the top left corner; then the HV of the underapproximation is the extent of the visible blue (top left) area.
The HV is the only known metric that is strictly monotonic~\cite{ZKT08}, meaning that if one approximation strictly contains another, its HV will be larger.
This makes it popular in multi-objective research~\cite{ZTJ15,LLC25,OTOS23,ZTLFF03}, but a disadvantage is that the reference point is arbitrary.
Thus the HV is useful for a qualitative comparison of results on the same model, but not across different models or for a quantitative evaluation.

To quantify the quality of an approximation w.r.t.\ a reference Pareto front/\allowbreak achievable set, we compute the multiplicative $\epsilon$~\cite{ZKT08} of the underapproximation.
This metric is the least factor by which each point in an approximation should be scalar-multiplied (or divided in case of a minimising objective) in order to be a superset of the reference.
A multiplicative $\epsilon$ of $1$ means that the approximation already fully contains the reference set and is thus optimal.
Infinite multiplicative $\epsilon$ values can arise when a minimising property needs to approximate the actual value $0$:
for a value $x \neq 0$, observe that only $\lim_{a \rightarrow \infty} \frac{x}{a} = 0$.
Since the multiplicative $\epsilon$ only considers the worst-performing point of the approximation, it can---and in our benchmarks does---happen that two approximations have equal multiplicative $\epsilon$ but one is a superset of the other.
Improvements in the rest of the approximation beyond the worst point can only be captured using the HV.

We run each experiment three times, each with a different strategy seed (of $1$, $2$, and $3$).
We report the average of the three HVs, and of the multiplicative $\epsilon$s if a reference front/set is known (either from insights into the model, or obtained via PMC).
We ran our experiments on a 64-bit Ubuntu Linux 22.04 system with AMD Ryzen~9 7950X3D CPU (4.2-5.7\,GHz) and 128\,GB~of~RAM.

\subsection{Comparison with PMC}

\begin{table}[t]
    \caption{Benchmarks instances and comparison to PMC.}
    \label{tab:benchmarks}
    \scriptsize
    \setlength{\tabcolsep}{2.5pt}
    \centering
    \begin{tabular}{>{\ttfamily}lcrrcc}
        \toprule
        {\normalfont benchmark instance} & objectives & states\phantom{**} & PMC time & our time & multiplicative $\epsilon$ \\
        \midrule
        care\_home.1-160 & 2 & 3.46\,M*\phantom{*} & ERR\phantom{**} & TO & -- \\
        energy\_aware\_scheduling.3-2 & 2 & 409\,k\phantom{**} & 7.6\,s\phantom{**} & [5.7\,s, 13\,s] & [6.84, 20.3] \\
        energy\_aware\_scheduling.3-3 & 2 & 1.86\,M\phantom{**} & 1.4\,min\phantom{**} & [5.7\,s, 13\,s] & [7.01, 19.9] \\
        energy\_aware\_scheduling.3-4 & 2 & 10.7\,M\phantom{**} & 15\,min\phantom{**} & [5.6\,s, 13\,s] & [7.30, 18.7] \\
        mars\_rover.10-1 & 2 & 151\,k\phantom{**} & 6.0\,s\phantom{**} & [1.6\,s, 11\,s] & [2.15, 3.55] \\
        mars\_rover.20-1 & 2 & 1.38\,M\phantom{**} & 1.3\,min\phantom{**} & [1.5\,s, 9.7\,s] & [15.4, 43.9] \\
        resource\_gathering.2-1-1-1 & 2 & 3.46\,k\phantom{**} & 0.1\,s*\phantom{*} & TO & -- \\
        \midrule
        breakable\_bottles.5-2 & 3 & 1.16\,k\phantom{**} & 0.1\,s\phantom{**} & [1.8\,s, 4.4\,s] & [1.29, 3.09] \\
        breakable\_bottles.10-100 & 3 & 36\,k\phantom{**} & 21\,s\phantom{**} & [7.3\,s, 15\,s] & [10.1, 32.0] \\
        breakable\_bottles.10-1000 & 3 & 40.6\,M\phantom{**} & TO\phantom{**} & [9.8\,s, 28\,s] & [43.5, 117] \\
        deep\_sea.5-100 & 2 & 2.40\,k\phantom{**} & 0.4\,s\phantom{**} & [1.7\,s, 3.5\,s] & [1.02, 1.34] \\
        deep\_sea.30-1000 & 2 & 532\,k\phantom{**} & 3.8\,min\phantom{**} & [2.4\,s, 22\,s] & [1.35, 1.71] \\
        deep\_sea.100-5000 & 2 & 26.1\,M\phantom{**} & TO\phantom{**} & [2.5\,s, 28\,s] & [1.71, 3.08] \\
        deep\_sea\_probabilistic.5-100 & 2 & 2.21\,k\phantom{**} & 0.5\,s\phantom{**} & [1.7\,s, 3.5\,s] & [1.00, 1.07] \\
        deep\_sea\_probabilistic.30-1000 & 2 & 1.04\,M\phantom{**} & 6.0\,min\phantom{**} & [2.7\,s, 12\,s] & [1.67, 2.26] \\
        deep\_sea\_probabilistic.100-5000 & 2 & 126\,M\phantom{**} & TO\phantom{**} & [2.7\,s, 12\,s] & -- \\
        fruit\_tree\_2.5 & 2 & 77\phantom{**} & 0.1\,s\phantom{**} & [1.2\,s, 17\,s] & [1.01, 1.07] \\
        fruit\_tree\_2.10 & 2 & 2.56\,k\phantom{**} & 1.0\,s\phantom{**} & [1.5\,s, 13\,s] & [1.01, 1.08] \\
        fruit\_tree\_2.30 & 2 & 2.7\,B** & TO\phantom{**} & [2.7\,s, 11\,s] & -- \\
        fruit\_tree\_6.7 & 6 & 317\phantom{**} & TO\phantom{**} & [10\,s, 8.4\,min] & [1.00, 1.42] \\
        fruit\_tree\_6.14 & 6 & 41.0\,k\phantom{**} & TO\phantom{**} & [5.4\,min, TO] & [1.31, 1.53] \\
        fruit\_tree\_6.30 & 6 & 2.7\,B** & TO\phantom{**} & [4.4\,min, TO] & -- \\
        puddle\_world\_barto\_big.30 & 2 & 411\,k*\phantom{*} & 10s*\phantom{*} & [4.2\,s, 7.9\,s] & [1.19, $\infty$) \\
        puddle\_world\_barto\_big.50 & 2 & 1.82\,M*\phantom{*} & 46s*\phantom{*} & [4.8\,s, 8.8\,s] & $\infty$ \\
        puddle\_world\_barto\_big.500 & 2 & 59.0\,M*\phantom{*} & TO*\phantom{*} & [6.0\,s, 20\,s]  & -- \\
        puddle\_world\_barto\_small.50 & 2 & 841\,k*\phantom{*} & 36\,s*\phantom{*} & [3.0\,s, 5.4\,s] & $\infty$ \\
        puddle\_world\_barto\_small.100 & 2 & 2.68\,M*\phantom{*} & 1.8\,min*\phantom{*} & [3.6\,s, 6.6\,s] & $\infty$ \\
        puddle\_world\_barto\_small.500 & 2 & 22.4\,M*\phantom{*} & 20\,min*\phantom{*} & [4.9\,s, 11\,s] & $\infty$ \\
        puddle\_world\_hansen.80 & 2 & 9.15\,M*\phantom{*} & 7.8\,min*\phantom{*} & [46\,s, 53\,s] & $\infty$ \\
        puddle\_world\_hansen.500 & 2 & 179\,M*\phantom{*} & TO*\phantom{*} & [52\,s, 67\,s] & -- \\
        puddle\_world\_hansen.1000 & 2 & 720\,M** & TO*\phantom{*} & [56\,s, 77\,s] & -- \\
        puddle\_world\_ring.60 & 2 & 5.41\,M*\phantom{*} & 5.0\,min*\phantom{*} & [7.0\,s, 12\,s] & $\infty$ \\
        puddle\_world\_ring.150 & 2 & 28.3\,M*\phantom{*} & 52\,min*\phantom{*} & [8.2\,s, 16\,s] & $\infty$ \\
        puddle\_world\_ring.500 & 2 & 128\,M*\phantom{*} & TO*\phantom{*} & [9.2\,s, 24\,s] & -- \\
        puddle\_world\_river.30 & 2 & 674\,k*\phantom{*} & 27\,s*\phantom{*} & [44\,s, 49\,s] & $\infty$ \\
        puddle\_world\_river.50  & 2 & 3.39\,M*\phantom{*} & 2.8\,min*\phantom{*} & [45\,s, 50\,s] & $\infty$ \\
        puddle\_world\_river.500 & 2 & 177\,M*\phantom{*} & TO*\phantom{*} & [55\,s, 69\,s] & -- \\
        puddle\_world\_tiny.5 & 2 & 1.27\,k*\phantom{*} & 0.4\,s*\phantom{*} & [1.2\,s, 4.4\,s] & [1.00, $\infty$) \\
        puddle\_world\_tiny.10 & 2 & 4.55\,k*\phantom{*} & 0.6\,s*\phantom{*} & [1.2\,s, 3.4\,s] & [1.00, $\infty$) \\
        puddle\_world\_tiny.100 & 2 & 131\,k*\phantom{*} & 2.3\,s*\phantom{*} & [1.3\,s, 7.1\,s] & [2.26, $\infty$) \\
        \bottomrule
    \end{tabular}
\end{table}

\Cref{tab:benchmarks} first lists the number of objectives in the property and the number of states of the MDP for each benchmark instance.
The latter has been computed via PMC (using \tool{Storm} in exact mode where possible, otherwise approximated via \mcsta's~\cite{HH15} new sound multi-objective PMC implementation~\cite{HQW26} and marked~*) or, if that failed, estimated by the authors from model-specific insights (marked~**).
It then compares the Pareto fronts---obtained via PMC (assumed to be correct and taken as the reference) or manually from model insights---with those approximated by our approaches, and the corresponding runtimes:
``PMC time'' is the runtime of \tool{Storm} or \mcsta, or a timeout ``TO'' ($>1$~hour). %
Column ``our time'' gives the interval from lowest to highest runtime of our approaches over all $3 \cdot 6 \cdot 3$ experiments performed per instance ($\text{TO} > 10$ minutes), while the last column contains analogous intervals over the multiplicative~$\epsilon$s.

We see that the PMC tools run out time %
on the larger RL benchmarks, while all our fixed-budget experiments succeed, often in seconds.
The SMC-based Pareto front underapproximations can get fairly close to the true Pareto fronts in this time, as observed for the \texttt{deep\_sea} and \texttt{fruit\_tree} models, but may also be quite far off as observed for the \texttt{breakable\_bottles} model.

\subsection{Fixed-Budget Algorithms Comparison}

We compare the fixed-budget algorithms and heuristics in terms of HV.
Each experiment performs a similar number of \emph{simulation runs}, but \emph{wall-clock runtimes}
may differ
due to per-strategy overhead in \modes
and because
some strategies induce longer runs.

\begin{wraptable}[10]{r}{6.25cm}
\vspace{-1em}
    \centering
    \caption{Config.-heuristic-alg.\ wins over all.}
    \vspace{8pt}
    \label{tab:results:overall}
    \scriptsize
    \setlength{\tabcolsep}{2pt}
    \begin{tabular}{cccccccccc}
        \toprule
         & \multicolumn{3}{c}{FIB} & \multicolumn{3}{c}{FSB} & \multicolumn{3}{c}{WVR}  \\
        \cmidrule(lr){2-4} \cmidrule(lr){5-7} \cmidrule(lr){8-10}
        Heuristic & C1 & C2 & C3 & C1 & C2 & C3 & C1 & C2 & C3 \\
        \midrule
        Simple & 0 & 1 & 1 & 2 & 2 & 5 & 0 & 0 & 4 \\
        FE & 0 & 1 & 0 & 0 & 0 & 8 & 0 & 0 & 0 \\
        FFE & 0 & 0 & 0 & 0 & 0 & 0 & 0 & 0 & 0 \\
        FFW & 0 & 1 & 0 & 0 & 0 & \textbf{9} & 0 & 0 & 0 \\
        FFEO & 1 & 0 & 0 & 0 & 0 & 0 & 0 & 0 & 0 \\
        CF & 0 & 0 & 0 & 0 & 0 & 2 & 0 & 0 & 0 \\
        \cmidrule(lr){1-1} \cmidrule(lr){2-4} \cmidrule(lr){5-7} \cmidrule(lr){8-10}
        Total & & 5 & & & \textbf{28} & & & 4 & \\
        \bottomrule
    \end{tabular}
\vspace{-3em}
\end{wraptable}

\paragraph{Best algorithm, heuristic, and configuration.}
No one combination of configuration, heuristic, and algorithm always wins, as shown in \Cref{tab:results:overall}, where we list the number of models on which each combination outperformed all other combinations (\ie delivered the highest HV, over the whole table).
FSB performs best overall. %
The total count is $37$ instead of $40$ since the experiments timed out on \texttt{care\_home} and \texttt{resource\_gathering}, while there is no unique best combination for the \texttt{deep\_sea\_probabilistic.5-100} instance.

\begin{table}[t]
    \parbox[t]{.33\linewidth}{
        \centering
        \setlength{\tabcolsep}{4pt}
        \caption{Heuristics.}
        \label{tab:results:heuristics}
        \vspace{-3pt}
        \scriptsize
        \begin{tabular}{cccc}
            \toprule
            Heuristic & FIB & FSB & WVR \\
            \midrule
            Simple & 6 & \textbf{16} & \textbf{26} \\
            FE & \textbf{9} & 8 & 1 \\
            FFE & 5 & 1 & 1 \\
            FFW & 4 & 9 & 3 \\
            FFEO & 8 & 1 & 1 \\
            CF & 3 & 2 & 0 \\
            \bottomrule
        \end{tabular}
    }%
    \hfill%
    \parbox[t]{.67\linewidth}{
        \centering
        \setlength{\tabcolsep}{4pt}
        \caption{Config.\ wins per heuristic-alg.\ combination.}
        \label{tab:results:config}
        \vspace{-5pt}
        \scriptsize
        \begin{tabular}{cccccccccc}
            \toprule
             & \multicolumn{3}{c}{FIB} & \multicolumn{3}{c}{FSB} & \multicolumn{3}{c}{WVR}  \\
            \cmidrule(lr){2-4} \cmidrule(lr){5-7} \cmidrule(lr){8-10}
            Heuristic & C1 & C2 & C3 & C1 & C2 & C3 & C1 & C2 & C3 \\
            \midrule
            Simple & 5 & 9 & \textbf{24} & 6 & 11 & \textbf{21} & 3 & 2 & \textbf{31} \\
            FE & 7 & 7 & \textbf{24} & 9 & 4 & \textbf{25} & 5 & 5 & \textbf{26} \\
            FFE & 7 & 5 & \textbf{26} & 9 & 6 & \textbf{23} & 6 & 8 & \textbf{22} \\
            FFW & 7 & 6 & \textbf{25} & 10 & 3 & \textbf{25} & 5 & 4 & \textbf{27} \\
            FFEO & 7 & 6 & \textbf{25} & 11 & 5 & \textbf{22} & 5 & 9 & \textbf{22} \\
            CF & \textbf{21} & 5 & 12 & 14 & \textbf{15} & 9 & 8 & 7 & \textbf{21} \\
            \bottomrule
        \end{tabular}
    }
\end{table}

The heuristics significantly influence the number of strategies surviving each iteration.
To evaluate their effect on each of the fixed-budget algorithms, we show in \Cref{tab:results:heuristics} the number of models for which each heuristic outperformed all others per algorithm (\ie within a column).
The simple strategy obtains the highest HV for most algorithms.
It thus appears that it is often beneficial (in our benchmark set) to discard strategies that do not look promising as early as possible and distribute the remaining budget over other strategies,
although as observed in \Cref{tab:results:overall}, other heuristics might lead to better overall results.

\Cref{tab:results:config} shows the number of models for which each configuration from~\Cref{tab:config} performs best for each heuristic-algorithm pair (\ie in a 3-column block within a row).
On these benchmarks, configuration $3$ performs best in almost every case, with the exception of the conservatively-far heuristic.
This heuristic needs the intervals to be smaller to discard suboptimal strategies and thus requires more simulation runs.
For the other approaches, it appears that having more, and thus presumably a wider variety of, strategies is more important than evaluating each strategy more precisely.

\begin{figure}[tb]
\centering
\scriptsize\hfill
\subfloat[Deep sea treasure]{
    \begin{tikzpicture}[scale=0.8]
        \begin{axis}[
            width=7.1cm,
            height=6.25cm,
            name=Axis,
            xlabel={Iteration},
            ylabel={Strategies in iteration},
            legend pos=north east,
            legend cell align=left,
            ymin=0,
            xmin=0,
            xticklabel style={/pgf/number format/fixed}
        ]
            \addplot[thick,blue,dotted,mark=none] table[x=iteration,y=deep_sea-P0-C0-Modest-0-S1] {data/FixedIterationBudget-Simple.dat};
            \addplot[thick,red,dashed,mark=none] table[x=iteration,y=deep_sea-P0-C0-Modest-0-S1] {data/FixedStrategyBudget-Simple.dat};
            \addplot[thick,orange,loosely dashdotted,mark=none] (0, 100) -- (10, 100) -- (10,4);
            \legend{Fixed iteration budget, Fixed strategy budget, Weight vector refinement}
        \end{axis}
    \end{tikzpicture}
    \label{fig:strategies:deep-sea}
}\hfill%
\subfloat[Fruit tree model]{
    \begin{tikzpicture}[scale=0.8]
        \begin{axis}[
            width=7.1cm,
            height=6.25cm,
            name=Axis,
            xlabel={Iteration},
            ylabel={Strategies in iteration},
            legend pos=south west,
            legend cell align=left,
            ymin=0,
            xmin=0,
            xticklabel style={/pgf/number format/fixed}
        ]
            \addplot[thick,blue,dotted,mark=none] table[x=iteration,y=exponential-P0-C0-Modest-0-S1] {data/FixedIterationBudget-Simple.dat};
            \addplot[thick,red,dashed,mark=none] table[x=iteration,y=exponential-P0-C0-Modest-0-S1] {data/FixedStrategyBudget-Simple.dat};
            \addplot[thick,orange,loosely dashdotted,mark=none] (0, 100) -- (10, 100) -- (10,58);
            \legend{Fixed iteration budget, Fixed strategy budget, Weight vector refinement}
        \end{axis}
    \end{tikzpicture}
    \label{fig:strategies:exponential}
}\hfill${}$
\caption{Strategies per iteration}
\label{fig:strategies}
\end{figure}

\paragraph{Strategies per iteration.}
The number of strategies remaining in each iteration naturally differs between the algorithms as illustrated for a \texttt{deep\_sea} and a \texttt{fruit\_tree} instance using the simple heuristic in \Cref{fig:strategies}.
In the \texttt{deep\_sea} model, most sampled strategies are suboptimal, whereas in the \texttt{fruit\_tree} model, most strategies are close to Pareto-optimal.
When most strategies are determined to be suboptimal, the FIB algorithm instantly discards most strategies.
FSB does the same, but uses the remaining budget to sample new strategies.
WVR, in contrast, only drops strategies after the last iteration.

\paragraph{Strategy seeds.}
As stated previously, each experiment was run with three different strategy seeds.
The reason for doing so can be seen in the underapproximations for the three different seeds for one of the racetrack-puddle world models in \Cref{fig:seed}:
Even though the model, algorithm, heuristic, and configuration are the same, there is still a significant difference in the computed underapproximations.

\subsection{The Incremental Scheme}

The effect of the incremental scheme's parameters can be seen without running experiments, as in \Cref{fig:approximation-settings:middle}.
Nevertheless, we experimentally inspect the underapproximations it finds, which we show for the deep sea treasure model and different timeouts in \Cref{fig:incremental:deep-sea}.
As expected, better strategies are sampled over time, and the approximation closes in on the true Pareto front.

\paragraphskip\noindent
Many more plots from our experiments, in particular visualisations like \Cref{fig:ParetoFrontExtendedExample} of the Pareto fronts found via our SMC-based approaches, are in \Cref{ap:AdditionalResults}.

\begin{figure}[b!]
\begin{minipage}[t]{0.33\textwidth}
    \centering
        \begin{tikzpicture}[scale=0.55]
            \begin{axis}[
                width=7.1cm,
                height=6.25cm,
                name=Axis,
                xlabel={$\bm{\leftarrow}$ minimise fuel consumed},
                ylabel={$\bm{\leftarrow}$ minimise puddle punishment},
                legend pos=north east,
                legend cell align=left,
                ymin=0,
                xmin=0.15,
                xmax=0.25,
                ymax=0.55,
                xticklabel style={/pgf/number format/fixed}
            ]
                \addplot[thick,blue!40!white] (0, 0) -- (0,0);
                \addplot[thick,red!40!white] (0, 0) -- (0,0);
                \addplot[thick,orange!40!white] (0, 0) -- (0,0);
                \legend{Seed 1, Seed 2, Seed 3}

                \addplot[blue!40!white, fill=blue!40!white,nearly transparent] (0.1660485752,0.55) --(0.1660485752,0.248044746119273) -- (0.16651497413983737,0.0725878707118558) -- (0.2135, 0.043) -- (0.1972857863167567,0.09646565385585537) -- (0.18966542609337694,0.1472948408695066) -- (0.1879609554,0.3865726815252643) -- (0.1879609554,0.55) -- cycle;
                \addplot[red!40!white, fill=red!40!white,nearly transparent] table[x=x-lower,y=y-lower, col sep=comma] {data/puddle_world_barto_small-P0-C1-Modest-FixedIterationBudget-Simple-1-S2.dat} -- (0.25, 0.0254297983) -- (0.25, 0.55) -- (0.1945592353, 0.55);
                \addplot[orange!40!white, fill=orange!40!white,nearly transparent] table[x=x-lower,y=y-lower, col sep=comma] {data/puddle_world_barto_small-P0-C1-Modest-FixedIterationBudget-Simple-1-S3.dat} -- (0.25, 0.0354577552) -- (0.25, 0.55) -- (0.1879609554, 0.55);

                \addplot[thick,blue!80!white,nearly transparent] (0.1660485752,0.55) --(0.1660485752,0.248044746119273) -- (0.16651497413983737,0.0725878707118558) -- (0.23351039662782314,0.0319690939) -- (0.25,0.0319690939);
                \addplot[thick,red!80!white,nearly transparent] (0.1945592353, 0.55) -- (0.1945592353, 0.49220873244461916) -- (0.2012569025514852, 0.0254297983) -- (0.25, 0.0254297983);
                \addplot[thick,orange!80!white,nearly transparent] (0.1879609554,0.55) --(0.1879609554,0.3865726815252643) -- (0.18966542609337694,0.1472948408695066) -- (0.1972857863167567,0.09646565385585537) -- (0.21604894455164045,0.0354577552) -- (0.25,0.0354577552);
            \end{axis}
        \end{tikzpicture}\vspace{-3pt}
        \caption{Strategy seeds.}
        \label{fig:seed}
\end{minipage}%
\begin{minipage}[t]{0.33\textwidth}
    \centering
        \begin{tikzpicture}[scale=0.55]
            \begin{axis}[
                width=7.1cm,
                height=6.25cm,
                name=Axis,
                xlabel={$\bm{\leftarrow}$ minimise fuel consumed},
                ylabel={maximise treasure found $\bm{\rightarrow}$},
                legend pos=south east,
                legend cell align=left,
                ymin=0,
                xmin=0,
                xmax=600,
                xticklabel style={/pgf/number format/fixed}
            ]
                \addplot[thick,blue!40!white] (0, 0) -- (0,0);
                \addplot[thick,red!40!white] (0, 0) -- (0,0);
                \addplot[thick,orange!40!white] (0, 0) -- (0,0);
                \addplot[thick] (0, 0) -- (0,0);
                \legend{10 seconds, 30 second, 60 seconds, Pareto front}

                \addplot[orange!40!white,fill=orange!40!white,nearly transparent] table[x=x-lower,y=y-lower, col sep=comma] {data/deep_sea-P0-C1-Modest-ParetoApproximation-2-S1.dat} -- (600, 7681) -- (600, 7280) -- (536, 7280) -- (281, 6633) -- (157, 5916) -- (81, 5000) -- (49, 4582) -- (29, 4000) (12, 3316) -- (6, 2449) -- (3, 1732) -- (1, 1000) -- cycle;
                \addplot[thick,orange!80!white,nearly transparent] table[x=x-lower,y=y-lower, col sep=comma] {data/deep_sea-P0-C1-Modest-ParetoApproximation-2-S1.dat} -- (600, 7681);

                \addplot[red!40!white,fill=red!40!white,nearly transparent] table[x=x-lower,y=y-lower, col sep=comma] {data/deep_sea-P0-C1-Modest-ParetoApproximation-1-S1.dat} -- (600, 7280) -- (600, 6633) -- (281, 6633) -- (157, 5916) -- (81, 5000) -- (66, 4690) -- (28, 3605) -- (20, 3316) -- (6, 2449) -- (3, 1732) -- (1, 1000) -- cycle;
                \addplot[thick,red!80!white,nearly transparent] table[x=x-lower,y=y-lower, col sep=comma] {data/deep_sea-P0-C1-Modest-ParetoApproximation-1-S1.dat} -- (600, 7280);

                \addplot[blue!80!white, fill=blue!40!white,nearly transparent] table[x=x-lower,y=y-lower, col sep=comma] {data/deep_sea-P0-C1-Modest-ParetoApproximation-0-S1.dat} -- (600,6633) -- (600, 0) -- (0, 0) -- cycle;
                \addplot[thick,black] table[x=x-correct,y=y-correct, col sep=comma] {data/deep_sea-P0-C1-Modest-ParetoApproximation-2-S1.dat};
            \end{axis}
        \end{tikzpicture}\vspace{-3pt}
        \caption{Incremental scheme.}
        \label{fig:incremental:deep-sea}
\end{minipage}%
\begin{minipage}[t]{0.33\textwidth}
\centering
\begin{tikzpicture}[scale=0.55]
    \begin{axis}[
        width=7.1cm,
        height=6.25cm,
        name=Axis,
        xlabel={$\bm{\leftarrow}$ minimise fuel},
        ylabel={maximise treasure $\bm{\rightarrow}$},
        xticklabel style={/pgf/number format/fixed}
    ]
        \addplot[thick, blue!80!white, fill=blue!40!white,nearly transparent] table[x=x_0,y=x_1, col sep=comma, discard if not={type}{under}] {data/appendix/deep_sea-P0-C1-Modest-FixedStrategyBudget-Simple-2-S1.dat} -- (\pgfkeysvalueof{/pgfplots/xmax}, 7548.999033856793) -- (\pgfkeysvalueof{/pgfplots/xmax}, \pgfkeysvalueof{/pgfplots/ymin}) -- (1.0000320148178616, \pgfkeysvalueof{/pgfplots/ymin}) -- cycle;
        \addplot[thick, red!80!white, fill=red!40!white,nearly transparent] table[x=x_0,y=x_1, col sep=comma, discard if not={type}{under}] {data/appendix/deep_sea-P0-C1-Modest-FixedStrategyBudget-Simple-2-S2.dat} -- (\pgfkeysvalueof{/pgfplots/xmax}, 7810.001275285751) -- (\pgfkeysvalueof{/pgfplots/xmax}, \pgfkeysvalueof{/pgfplots/ymin}) -- (1.0000320148178616, \pgfkeysvalueof{/pgfplots/ymin}) -- cycle;
        \addplot[thick, orange!80!white, fill=orange!40!white,nearly transparent] table[x=x_0,y=x_1, col sep=comma, discard if not={type}{under}] {data/appendix/deep_sea-P0-C1-Modest-FixedStrategyBudget-Simple-2-S3.dat} -- (\pgfkeysvalueof{/pgfplots/xmax}, 7810.0002926848465) -- (\pgfkeysvalueof{/pgfplots/xmax}, \pgfkeysvalueof{/pgfplots/ymin}) -- (1.0000320148178616, \pgfkeysvalueof{/pgfplots/ymin}) -- cycle;
        \addplot[thick] table[x=x_0,y=x_1, col sep=comma, discard if not={type}{ref}] {data/appendix/deep_sea-P0-C1-Modest-FixedStrategyBudget-Simple-2-S1.dat};
    \end{axis}
\end{tikzpicture}\vspace{-3pt}
\caption{\texttt{deep\_sea.10-100}}
\label{fig:ParetoFrontExtendedExample}
\end{minipage}
\end{figure}

\section{Conclusion}
\label{sec:Conclusion}

\noindent
We introduced the first multi-objective SMC approach that
(i)~addresses the general case of $d$ equal-priority probabilistic objectives and
(ii)~is constant-memory \wrt the model's state space size.
Our algorithms approximate the achievable set and Pareto front.
They provide statistical guarantees on the under\-approximation, and the incremental approach also comes with a limit guarantee for the overapproximation.
Our experiments indicate that, in finite time, the fixed strategy budget algorithm with the simple heuristic, a high number of initial strategies, and a small number of runs for each tends to perform best.

Our approach directly applies to models beyond MDP for which single-objective LSS is available.
It also extends to infinite-state models in principle, as long as runs almost surely terminate; however, the overapproximation in the incremental approach would no longer converge because we could no longer assume every strategy to eventually be sampled given fixed-size strategy identifiers.

\paragraph{Data availability.}
The exercised versions of \modes and \tool{Storm}, benchmarks, produced results, and a docker environment are available on \tool{Zenodo}~\cite{AHWW26}.

\bibliographystyle{splncs04}
\bibliography{paper}

\begin{thebibliography}{100}
\providecommand{\url}[1]{\texttt{#1}}
\providecommand{\urlprefix}{URL }
\providecommand{\doi}[1]{https://doi.org/#1}

\bibitem{Abd07}
Abdi, H.: The {B}onferonni and {Š}idák corrections for multiple comparisons.
  In: Salkind, N.J. (ed.) Encyclopedia of measurement and statistics. Sage
  Publications (2007),
  \url{https://personal.utdallas.edu/~herve/Abdi-Bonferroni2007-pretty.pdf}

\bibitem{AP18}
Agha, G., Palmskog, K.: A survey of statistical model checking. {ACM} Trans.
  Model. Comput. Simul.  \textbf{28}(1),  6:1--6:39 (2018).
  \doi{10.1145/3158668}

\bibitem{ADLR22}
Akraoui, B.E., Daoui, C., Larach, A., Rahhali, K.: Decomposition methods for
  solving finite-horizon large {MDP}s. J. Math.  \textbf{2022} (2022).
  \doi{10.1155/2022/8404716}

\bibitem{ABBC+24}
Andriushchenko, R., Bork, A., Budde, C.E., Češka, M., Grover, K., Hahn, E.M.,
  Hartmanns, A., Israelsen, B., Jansen, N., Jeppson, J., Junges, S., Köhl,
  M.A., Könighofer, B., Křetínský, J., Meggendorfer, T., Parker, D.,
  Pranger, S., Quatmann, T., Ruijters, E., Taylor, L., Volk, M., Weininger, M.,
  Zhang, Z.: Tools at the frontiers of quantitative verification ({QC}omp 2023
  competition report). In: Beyer, D., Hartmanns, A., Kordon, F. (eds.)
  {TOOL}ympics Challenge 2023. Lecture Notes in Computer Science, vol. 14550,
  pp. 90--146. Springer (2024). \doi{10.1007/978-3-031-67695-6_4}

\bibitem{ADKW20}
Ashok, P., Daca, P., Kret{\'{\i}}nsk{\'{y}}, J., Weininger, M.: Statistical
  model checking: Black or white? In: Margaria, T., Steffen, B. (eds.) 9th
  International Symposium on Leveraging Applications of Formal Methods ({ISoLA}
  2020). Lecture Notes in Computer Science, vol. 12476, pp. 331--349. Springer
  (2020). \doi{10.1007/978-3-030-61362-4_19}

\bibitem{ACOD16}
Auer, P., Chiang, C.K., Ortner, R., Drugan, M.M.: {P}areto front identification
  from stochastic bandit feedback. In: Gretton, A., Robert, C.C. (eds.) 19th
  International Conference on Artificial Intelligence and Statistics ({AISTATS}
  2016). {JMLR} Workshop and Conference Proceedings, vol.~51, pp. 939--947.
  JMLR.org (2016), \url{http://proceedings.mlr.press/v51/auer16.html}

\bibitem{AMAD+25}
Awadallah, M.A., Makhadmeh, S.N., Al-Betar, M.A., Dalbah, L.M., Al-Redhaei, A.,
  Kouka, S., Enshassi, O.S.: Multi-objective ant colony optimization: Review.
  Arch. Comput. Methods Eng.  \textbf{32},  995--1037 (2025).
  \doi{10.1007/s11831-024-10178-4}

\bibitem{Bai16}
Baier, C.: Probabilistic model checking. In: Esparza, J., Grumberg, O.,
  Sickert, S. (eds.) Dependable Software Systems Engineering, {NATO} Science
  for Peace and Security Series -- {D}: Information and Communication Security,
  vol.~45, pp. 1--23. {IOS} Press (2016). \doi{10.3233/978-1-61499-627-9-1}

\bibitem{BAFK18}
Baier, C., de~Alfaro, L., Forejt, V., Kwiatkowska, M.: Model checking
  probabilistic systems. In: Clarke, E.M., Henzinger, T.A., Veith, H., Bloem,
  R. (eds.) Handbook of Model Checking, pp. 963--999. Springer (2018).
  \doi{10.1007/978-3-319-10575-8_28}

\bibitem{BCGG+20}
Baier, C., Christakis, M., Gros, T.P., Gro{\ss}, D., Gumhold, S., Hermanns, H.,
  Hoffmann, J., Klauck, M.: Lab conditions for research on explainable
  automated decisions. In: Heintz, F., Milano, M., O'Sullivan, B. (eds.) 1st
  International Workshop on Trustworthy {AI} -- Integrating Learning,
  Optimization and Reasoning ({TAILOR} 2020). Lecture Notes in Computer
  Science, vol. 12641, pp. 83--90. Springer (2020).
  \doi{10.1007/978-3-030-73959-1_8}

\bibitem{BKLPW17}
Baier, C., Klein, J., Leuschner, L., Parker, D., Wunderlich, S.: Ensuring the
  reliability of your model checker: Interval iteration for {M}arkov decision
  processes. In: Majumdar, R., Kuncak, V. (eds.) 29th International Conference
  on Computer Aided Verification ({CAV} 2017). Lecture Notes in Computer
  Science, vol. 10426, pp. 160--180. Springer (2017).
  \doi{10.1007/978-3-319-63387-9\_8}

\bibitem{BBS95}
Barto, A.G., Bradtke, S.J., Singh, S.P.: Learning to act using real-time
  dynamic programming. Artif. Intell.  \textbf{72}(1-2),  81--138 (1995).
  \doi{10.1016/0004-3702(94)00011-O}

\bibitem{BCDF+07}
Behrmann, G., Cougnard, A., David, A., Fleury, E., Larsen, K.G., Lime, D.:
  {UPPAAL-Tiga}: Time for playing games! In: Damm, W., Hermanns, H. (eds.) 19th
  International Conference on Computer Aided Verification ({CAV} 2007). Lecture
  Notes in Computer Science, vol.~4590, pp. 121--125. Springer (2007).
  \doi{10.1007/978-3-540-73368-3_14}

\bibitem{Bel57}
Bellman, R.: A {M}arkovian decision process. Journal of Mathematics and
  Mechanics  \textbf{6}(5),  679--684 (1957)

\bibitem{BGHK+19}
Bisgaard, M., Gerhardt, D., Hermanns, H., Krc{\'{a}}l, J., Nies, G., Stenger,
  M.: Battery-aware scheduling in low orbit: the {GomX-3} case. Formal Aspects
  Comput.  \textbf{31}(2),  261--285 (2019). \doi{10.1007/S00165-018-0458-2}

\bibitem{BM94}
Boyan, J.A., Moore, A.W.: Generalization in reinforcement learning: Safely
  approximating the value function. In: Tesauro, G., Touretzky, D.S., Leen,
  T.K. (eds.) Advances in Neural Information Processing Systems 7 ({NIPS}
  1994). pp. 369--376. {MIT} Press (1994),
  \url{https://proceedings.neurips.cc/paper\_files/paper/1994/hash/ef50c335cca9f340bde656363ebd02fd-Abstract.html}

\bibitem{BDH24}
Budde, C.E., D'Argenio, P.R., Hartmanns, A.: Digging for decision trees: A case
  study in strategy sampling and learning. In: Steffen, B. (ed.) 2nd
  International Conference on Bridging the Gap Between {AI} and Reality
  ({AISoLA} 2024). Lecture Notes in Computer Science, vol. 15217, pp. 354--378.
  Springer (2024). \doi{10.1007/978-3-031-75434-0_24}

\bibitem{BDHS20}
Budde, C.E., D'Argenio, P.R., Hartmanns, A., Sedwards, S.: An efficient
  statistical model checker for nondeterminism and rare events. Int. J. Softw.
  Tools Technol. Transf.  \textbf{22}(6),  759--780 (2020).
  \doi{10.1007/S10009-020-00563-2}

\bibitem{BHMWW25a}
Budde, C.E., Hartmanns, A., Meggendorfer, T., Weininger, M., Wienh{\"{o}}ft,
  P.: Sound statistical model checking for probabilities and expected rewards.
  In: Gurfinkel, A., Heule, M. (eds.) 31st International Conference on Tools
  and Algorithms for the Construction and Analysis of Systems ({TACAS} 2025).
  Lecture Notes in Computer Science, vol. 15696, pp. 167--190. Springer (2025).
  \doi{10.1007/978-3-031-90643-5_9}

\bibitem{BHMWW25b}
Budde, C.E., Hartmanns, A., Meggendorfer, T., Weininger, M., Wienh{\"{o}}ft,
  P.: Statistical model checking beyond means: Quantiles, {CVaR}, and the {DKW}
  inequality. In: 2nd International Joint Conference on Quantitative Evaluation
  of Systems and Formal Modeling and Analysis of Timed Systems ({QEST+FORMATS}
  2025). Lecture Notes in Computer Science, vol. 16143, pp. 83--94. Springer
  (2025). \doi{10.1007/978-3-032-05792-1_5}

\bibitem{CWG20}
Chen, D., Wang, Y., Gao, W.: Combining a gradient-based method and an evolution
  strategy for multi-objective reinforcement learning. Appl. Intell.
  \textbf{50}(10),  3301--3317 (2020). \doi{10.1007/S10489-020-01702-7}

\bibitem{CEHH+21}
Christakis, M., Eniser, H.F., Hermanns, H., Hoffmann, J., Kothari, Y., Li, J.,
  Navas, J.A., W{\"{u}}stholz, V.: Automated safety verification of programs
  invoking neural networks. In: Silva, A., Leino, K.R.M. (eds.) 33rd
  International Conference on Computer Aided Verification ({CAV} 2021). Lecture
  Notes in Computer Science, vol. 12759, pp. 201--224. Springer (2021).
  \doi{10.1007/978-3-030-81685-8_9}

\bibitem{CP34}
Clopper, C., Pearson, E.S.: The use of confidence or fiducial limits
  illustrated in the case of the binomial. Biometrika  \textbf{26}(4),
  404--413 (1934). \doi{10.1093/biomet/26.4.404}

\bibitem{DFH20}
D'Argenio, P.R., Fraire, J.A., Hartmanns, A.: Sampling distributed schedulers
  for resilient space communication. In: Lee, R., Jha, S., Mavridou, A. (eds.)
  12th International {NASA} Formal Methods Symposium ({NFM} 2020). Lecture
  Notes in Computer Science, vol. 12229, pp. 291--310. Springer (2020).
  \doi{10.1007/978-3-030-55754-6_17}

\bibitem{DGHS18}
D'Argenio, P.R., Gerhold, M., Hartmanns, A., Sedwards, S.: A hierarchy of
  scheduler classes for stochastic automata. In: Baier, C., Lago, U.D. (eds.)
  21st International Conference on Foundations of Software Science and
  Computation Structures ({FOSSACS} 2018). Lecture Notes in Computer Science,
  vol. 10803, pp. 384--402. Springer (2018). \doi{10.1007/978-3-319-89366-2_21}

\bibitem{DHLS16}
D'Argenio, P.R., Hartmanns, A., Legay, A., Sedwards, S.: Statistical
  approximation of optimal schedulers for probabilistic timed automata. In:
  {\'{A}}brah{\'{a}}m, E., Huisman, M. (eds.) 12th International Conference on
  Integrated Formal Methods ({iFM} 2016). Lecture Notes in Computer Science,
  vol.~9681, pp. 99--114. Springer (2016). \doi{10.1007/978-3-319-33693-0_7}

\bibitem{AHWW26}
D'Argenio, P.R., Hartmanns, A., Wienh{\"{o}}ft, P., van Wijk, M.:
  Multi-objective statistical model checking using lightweight strategy
  sampling (artifact) (2026). \doi{10.5281/zenodo.19660159}

\bibitem{DLST15}
D'Argenio, P.R., Legay, A., Sedwards, S., Traonouez, L.M.: Smart sampling for
  lightweight verification of {M}arkov decision processes. Int. J. Softw. Tools
  Technol. Transf.  \textbf{17}(4),  469--484 (2015).
  \doi{10.1007/S10009-015-0383-0}

\bibitem{DJLL+14}
David, A., Jensen, P.G., Larsen, K.G., Legay, A., Lime, D., S{\o}rensen, M.G.,
  Taankvist, J.H.: On time with minimal expected cost! In: Cassez, F., Raskin,
  J.F. (eds.) 12th International Symposium on Automated Technology for
  Verification and Analysis ({ATVA} 2014). Lecture Notes in Computer Science,
  vol.~8837, pp. 129--145. Springer (2014). \doi{10.1007/978-3-319-11936-6_10}

\bibitem{DJLMT15}
David, A., Jensen, P.G., Larsen, K.G., Mikučionis, M., Taankvist, J.H.:
  {U}ppaal {S}tratego. In: Baier, C., Tinelli, C. (eds.) 21st International
  Conference on Tools and Algorithms for the Construction and Analysis of
  Systems ({TACAS} 2015). Lecture Notes in Computer Science, vol.~9035, pp.
  206--211. Springer (2015). \doi{10.1007/978-3-662-46681-0_16}

\bibitem{DLLMP15}
David, A., Larsen, K.G., Legay, A., Mikučionis, M., Poulsen, D.B.: {U}ppaal
  {SMC} tutorial. Int. J. Softw. Tools Technol. Transf.  \textbf{17}(4),
  397--415 (2015). \doi{10.1007/s10009-014-0361-y}

\bibitem{DAPM02}
Deb, K., Agrawal, S., Pratap, A., Meyarivan, T.: A fast and elitist
  multiobjective genetic algorithm: {NSGA-II}. IEEE Trans. Evol. Comput.
  \textbf{6}(2),  182--197 (2002). \doi{10.1109/4235.996017}

\bibitem{EHZ10}
Eisentraut, C., Hermanns, H., Zhang, L.: On probabilistic automata in
  continuous time. In: 25th Annual {IEEE} Symposium on Logic in Computer
  Science ({LICS} 2010). pp. 342--351. {IEEE} Computer Society (2010).
  \doi{10.1109/LICS.2010.41}

\bibitem{EKVY08}
Etessami, K., Kwiatkowska, M.Z., Vardi, M.Y., Yannakakis, M.: Multi-objective
  model checking of {M}arkov decision processes. Log. Methods Comput. Sci.
  \textbf{4}(4) (2008). \doi{10.2168/LMCS-4(4:8)2008}

\bibitem{FANB+23}
Felten, F., Alegre, L.N., Now{\'{e}}, A., Bazzan, A.L.C., Talbi, E.G., Danoy,
  G., da~Silva, B.C.: A toolkit for reliable benchmarking and research in
  multi-objective reinforcement learning. In: Oh, A., Naumann, T., Globerson,
  A., Saenko, K., Hardt, M., Levine, S. (eds.) 37th Annual Conference on
  Advances in Neural Information Processing Systems ({NeurIPS} 2023) (2023),
  \url{http://papers.nips.cc/paper\_files/paper/2023/hash/4aa8891583f07ae200ba07843954caeb-Abstract-Datasets\_and\_Benchmarks.html}

\bibitem{FGR22}
Fickert, M., Gu, T., Ruml, W.: New results in bounded-suboptimal search. In:
  36th {AAAI} Conference on Artificial Intelligence ({AAAI} 2022). pp.
  10166--10173. {AAAI} Press (2022). \doi{10.1609/AAAI.V36I9.21256}

\bibitem{FKNP11}
Forejt, V., Kwiatkowska, M.Z., Norman, G., Parker, D.: Automated verification
  techniques for probabilistic systems. In: Bernardo, M., Issarny, V. (eds.)
  11th International School on Formal Methods for the Design of Computer,
  Communication and Software Systems ({SFM} 2011). Lecture Notes in Computer
  Science, vol.~6659, pp. 53--113. Springer (2011).
  \doi{10.1007/978-3-642-21455-4_3}

\bibitem{FKNPQ11}
Forejt, V., Kwiatkowska, M.Z., Norman, G., Parker, D., Qu, H.: Quantitative
  multi-objective verification for probabilistic systems. In: Abdulla, P.A.,
  Leino, K.R.M. (eds.) 17th International Conference on Tools and Algorithms
  for the Construction and Analysis of Systems ({TACAS} 2011). Lecture Notes in
  Computer Science, vol.~6605, pp. 112--127. Springer (2011).
  \doi{10.1007/978-3-642-19835-9_11}

\bibitem{FKP12}
Forejt, V., Kwiatkowska, M.Z., Parker, D.: {P}areto curves for probabilistic
  model checking. In: Chakraborty, S., Mukund, M. (eds.) 10th International
  Symposium on Automated Technology for Verification and Analysis ({ATVA}
  2012). Lecture Notes in Computer Science, vol.~7561, pp. 317--332. Springer
  (2012). \doi{10.1007/978-3-642-33386-6_25}

\bibitem{FHLSSTZ22}
Fu, C., Hahn, E.M., Li, Y., Schewe, S., Sun, M., Turrini, A., Zhang, L.: {EPMC}
  gets knowledge in multi-agent systems. In: Finkbeiner, B., Wies, T. (eds.)
  23rd International Conference on Verification, Model Checking, and Abstract
  Interpretation ({VMCAI} 2022). Lecture Notes in Computer Science, vol. 13182,
  pp. 93--107. Springer (2022). \doi{10.1007/978-3-030-94583-1_5}

\bibitem{Gar73}
Gardner, M.: Mathematical games. Sci. Am.  \textbf{228}(1),  108--115 (1973).
  \doi{10.1038/scientificamerican0173-108}

\bibitem{GHHKS23}
Gros, T.P., Hermanns, H., Hoffmann, J., Klauck, M., Steinmetz, M.: Analyzing
  neural network behavior through deep statistical model checking. Int. J.
  Softw. Tools Technol. Transf.  \textbf{25}(3),  407--426 (2023).
  \doi{10.1007/S10009-022-00685-9}

\bibitem{HM14}
Haddad, S., Monmege, B.: Reachability in {MDP}s: Refining convergence of value
  iteration. In: Ouaknine, J., Potapov, I., Worrell, J. (eds.) 8th
  International Workshop on Reachability Problems ({RP} 2014). Lecture Notes in
  Computer Science, vol.~8762, pp. 125--137. Springer (2014).
  \doi{10.1007/978-3-319-11439-2_10}

\bibitem{HPSSTW21}
Hahn, E.M., Perez, M., Schewe, S., Somenzi, F., Trivedi, A., Wojtczak, D.:
  Model-free reinforcement learning for lexicographic omega-regular objectives.
  In: Huisman, M., Pasareanu, C.S., Zhan, N. (eds.) 24th International Formal
  Methods Symposium ({FM} 2021). Lecture Notes in Computer Science, vol. 13047,
  pp. 142--159. Springer (2021). \doi{10.1007/978-3-030-90870-6\_8}

\bibitem{HH14}
Hartmanns, A., Hermanns, H.: The {M}odest {T}oolset: An integrated environment
  for quantitative modelling and verification. In: {\'{A}}brah{\'{a}}m, E.,
  Havelund, K. (eds.) 20th International Conference on Tools and Algorithms for
  the Construction and Analysis of Systems ({TACAS} 2014). Lecture Notes in
  Computer Science, vol.~8413, pp. 593--598. Springer (2014).
  \doi{10.1007/978-3-642-54862-8_51}

\bibitem{HH15}
Hartmanns, A., Hermanns, H.: Explicit model checking of very large {MDP} using
  partitioning and secondary storage. In: Finkbeiner, B., Pu, G., Zhang, L.
  (eds.) 13th International Symposium on Automated Technology for Verification
  and Analysis ({ATVA} 2015). Lecture Notes in Computer Science, vol.~9364, pp.
  131--147. Springer (2015). \doi{10.1007/978-3-319-24953-7\_10}

\bibitem{HJQW23}
Hartmanns, A., Junges, S., Quatmann, T., Weininger, M.: A practitioner's guide
  to {MDP} model checking algorithms. In: Sankaranarayanan, S., Sharygina, N.
  (eds.) 29th International Conference on Tools and Algorithms for the
  Construction and Analysis of Systems ({TACAS} 2023). Lecture Notes in
  Computer Science, vol. 13993, pp. 469--488. Springer (2023).
  \doi{10.1007/978-3-031-30823-9_24}

\bibitem{HQW26}
Hartmanns, A., Quatmann, T., van Wijk, M.: Tools and algorithms for sound
  multi-objective probabilistic model checking. In: 27th International
  Symposium on Formal Methods ({FM} 2026). Lecture Notes in Computer Science,
  Springer (2026), to appear.

\bibitem{HSD17}
Hartmanns, A., Sedwards, S., D'Argenio, P.R.: Efficient simulation-based
  verification of probabilistic timed automata. In: 2017 Winter Simulation
  Conference ({WSC} 2017). pp. 1419--1430. {IEEE} (2017).
  \doi{10.1109/WSC.2017.8247885}

\bibitem{HLIS+19}
Hasan, M.M., Lwin, K.T., Imani, M., Shabut, A.M., Bittencourt, L.F., Hossain,
  M.A.: Dynamic multi-objective optimisation using deep reinforcement learning:
  benchmark, algorithm and an application to identify vulnerable zones based on
  water quality. Eng. Appl. Artif. Intell.  \textbf{86},  107--135 (2019).
  \doi{10.1016/J.ENGAPPAI.2019.08.014}

\bibitem{HJLS23}
Hasrat, I.R., Jensen, P.G., Larsen, K.G., Srba, J.: A toolchain for domestic
  heat-pump control using {U}ppaal {S}tratego. Sci. Comput. Program.
  \textbf{230},  102987 (2023). \doi{10.1016/J.SCICO.2023.102987}

\bibitem{HRBK+22}
Hayes, C.F., Rădulescu, R., Bargiacchi, E., Källström, J., Macfarlane, M.,
  Reymond, M., Verstraeten, T., Zintgraf, L.M., Dazeley, R., Heintz, F.,
  Howley, E., Irissappane, A.A., Mannion, P., Nowé, A., Ramos, G., Restelli,
  M., Vamplew, P., Roijers, D.M.: A practical guide to multi-objective
  reinforcement learning and planning. Auton. Agents Multi Agent Syst.
  \textbf{36}(26) (2022). \doi{10.1007/s10458-022-09552-y}

\bibitem{HJKQV22}
Hensel, C., Junges, S., Katoen, J.P., Quatmann, T., Volk, M.: The probabilistic
  model checker {S}torm. Int. J. Softw. Tools Technol. Transf.  \textbf{24}(4),
   589--610 (2022). \doi{10.1007/S10009-021-00633-Z}

\bibitem{How60}
Howard, R.A.: Dynamic Programming and {M}arkov Processes. MIT Press (1960)

\bibitem{IGIT03}
Ito, K., Gofuku, A., Imoto, Y., Takeshita, M.: A study of reinforcement
  learning with knowledge sharing for distributed autonomous system. In: 2003
  {IEEE} International Symposium on Computational Intelligence in Robotics and
  Automation ({CIRA} 2003). pp. 1120--1125. {IEEE} (2003).
  \doi{10.1109/CIRA.2003.1222154}

\bibitem{JKP21}
Jain, A., Khetarpal, K., Precup, D.: Safe option-critic: learning safety in the
  option-critic architecture. Knowl. Eng. Rev.  \textbf{36}, ~e4 (2021).
  \doi{10.1017/S0269888921000035}

\bibitem{Kat16}
Katoen, J.P.: The probabilistic model checking landscape. In: Grohe, M.,
  Koskinen, E., Shankar, N. (eds.) 31st Annual {ACM/IEEE} Symposium on Logic in
  Computer Science ({LICS} 2016). pp. 31--45. {ACM} (2016).
  \doi{10.1145/2933575.2934574}

\bibitem{KSBD15}
Kr{\"{a}}hmann, D., Schubert, J., Baier, C., Dubslaff, C.: Ratio and weight
  quantiles. In: Italiano, G.F., Pighizzini, G., Sannella, D. (eds.) 40th
  International Symposium on Mathematical Foundations of Computer Science
  ({MFCS} 2015). Lecture Notes in Computer Science, vol.~9234, pp. 344--356.
  Springer (2015). \doi{10.1007/978-3-662-48057-1_27}

\bibitem{KM18}
Kret{\'{\i}}nsk{\'{y}}, J., Meggendorfer, T.: Conditional value-at-risk for
  reachability and mean payoff in {M}arkov decision processes. In: Dawar, A.,
  Gr{\"{a}}del, E. (eds.) 33rd Annual {ACM/IEEE} Symposium on Logic in Computer
  Science ({LICS} 2018). pp. 609--618. {ACM} (2018).
  \doi{10.1145/3209108.3209176}

\bibitem{KNP11}
Kwiatkowska, M., Norman, G., Parker, D.: {PRISM} 4.0: Verification of
  probabilistic real-time systems. In: Gopalakrishnan, G., Qadeer, S. (eds.)
  23rd International Conference on Computer Aided Verification ({CAV} 2011).
  Lecture Notes in Computer Science, vol.~6806, pp. 585--591. Springer (2011).
  \doi{10.1007/978-3-642-22110-1_47}

\bibitem{KNSS02}
Kwiatkowska, M.Z., Norman, G., Segala, R., Sproston, J.: Automatic verification
  of real-time systems with discrete probability distributions. Theor. Comput.
  Sci.  \textbf{282}(1),  101--150 (2002). \doi{10.1016/S0304-3975(01)00046-9}

\bibitem{LL16}
Larsen, K.G., Legay, A.: Statistical model checking: Past, present, and future.
  In: Margaria, T., Steffen, B. (eds.) 7th International Symposium on
  Leveraging Applications of Formal Methods ({ISoLA} 2016). Lecture Notes in
  Computer Science, vol.~9952, pp. 3--15. Springer (2016).
  \doi{10.1007/978-3-319-47166-2_1}

\bibitem{LLTYSG19}
Legay, A., Lukina, A., Traonouez, L.M., Yang, J., Smolka, S.A., Grosu, R.:
  Statistical model checking. In: Steffen, B., Woeginger, G.J. (eds.) Computing
  and Software Science -- State of the Art and Perspectives, Lecture Notes in
  Computer Science, vol. 10000, pp. 478--504. Springer (2019).
  \doi{10.1007/978-3-319-91908-9_23}

\bibitem{LST14}
Legay, A., Sedwards, S., Traonouez, L.M.: Scalable verification of {M}arkov
  decision processes. In: Canal, C., Idani, A. (eds.) 4th Workshop on Formal
  Methods in the Development of Software ({WS-FMDS} 2014). Lecture Notes in
  Computer Science, vol.~8938, pp. 350--362. Springer (2014).
  \doi{10.1007/978-3-319-15201-1_23}

\bibitem{LLC25}
Li, L., Li, G., Cai, G.: A local {P}areto front estimation framework for
  multi-objective optimization. In: Shen, W., Abel, M.H., Matta, N.,
  Barth{\`{e}}s, J.P.A., Luo, J., Zhang, J., Zhu, H., Peng, K. (eds.) 28th
  International Conference on Computer Supported Cooperative Work in Design
  ({CSCWD} 2025). pp. 2103--2108. {IEEE} (2025).
  \doi{10.1109/CSCWD64889.2025.11033349}

\bibitem{LTZCX24}
Li, Z., Tang, J., Zhao, H., Chen, C., Xie, S.: Dictionary learning-structured
  reinforcement learning with adaptive-sparsity regularizer. {IEEE} Trans.
  Aerosp. Electron. Syst.  \textbf{60}(2),  1753--1769 (2024).
  \doi{10.1109/TAES.2023.3342794}

\bibitem{Mil66}
Miller, R.G.: Simultaneous Statistical Inference. Springer (1966).
  \doi{10.1007/978-1-4613-8122-8}

\bibitem{MDN13}
Moffaert, K.V., Drugan, M.M., Now{\'{e}}, A.: Scalarized multi-objective
  reinforcement learning: Novel design techniques. In: 2013 {IEEE} Symposium on
  Adaptive Dynamic Programming and Reinforcement Learning ({ADPRL} 2013). pp.
  191--199. {IEEE} (2013). \doi{10.1109/ADPRL.2013.6615007}

\bibitem{NRR14}
Nguyen, A.T., Reiter, S., Rigo, P.: A review on simulation-based optimization
  methods applied to building performance analysis. Appl. Energy  \textbf{113},
   1043--1058 (2014). \doi{10.1016/j.apenergy.2013.08.061}

\bibitem{NNVN+20}
Nguyen, T.T., Nguyen, N.D., Vamplew, P., Nahavandi, S., Dazeley, R., Lim, C.P.:
  A multi-objective deep reinforcement learning framework. Eng. Appl. Artif.
  Intell.  \textbf{96},  103915 (2020). \doi{10.1016/J.ENGAPPAI.2020.103915}

\bibitem{Oka59}
Okamoto, M.: Some inequalities relating to the partial sum of binomial
  probabilities. Annals of the Institute of Statistical Mathematics
  \textbf{10}(1),  29--35 (1959)

\bibitem{OTOS23}
Okumura, N., Takagi, T., Ohta, Y., Sato, H.: {P}areto front upconvert on
  multi-objective building facility control optimization. In: Silva, S.,
  Paquete, L. (eds.) Genetic and Evolutionary Computation Conference ({GECCO}
  2023). pp. 1963--1971. {ACM} (2023). \doi{10.1145/3583133.3596339}

\bibitem{PV02}
Parsopoulos, K.E., Vrahatis, M.N.: Particle swarm optimization method in
  multiobjective problems. In: Lamont, G.B., Haddad, H., Papadopoulos, G.A.,
  Panda, B. (eds.) 2002 {ACM} Symposium on Applied Computing ({SAC} 2002). pp.
  603--607. {ACM} (2002). \doi{10.1145/508791.508907}

\bibitem{Pnu77}
Pnueli, A.: The temporal logic of programs. In: 18th Annual Symposium on
  Foundations of Computer Science ({FOCS} 1977). pp. 46--57. {IEEE} Computer
  Society (1977). \doi{10.1109/SFCS.1977.32}

\bibitem{Put94}
Puterman, M.L.: {M}arkov Decision Processes: Discrete Stochastic Dynamic
  Programming. Wiley Series in Probability and Statistics, Wiley (1994).
  \doi{10.1002/9780470316887}

\bibitem{Qua23}
Quatmann, T.: Verification of multi-objective {M}arkov models. Ph.D. thesis,
  RWTH Aachen University (2023). \doi{10.18154/RWTH-2023-09669}

\bibitem{QAVO19}
Quiroz, E.A.P., Apolin{\'{a}}rio, H.C.F., Villacorta, K.D.V., Oliveira, P.R.: A
  linear scalarization proximal point method for quasiconvex multiobjective
  minimization. J. Optim. Theory Appl.  \textbf{183}(3),  1028--1052 (2019).
  \doi{10.1007/S10957-019-01582-Z}

\bibitem{RRS17}
Randour, M., Raskin, J.F., Sankur, O.: Percentile queries in multi-dimensional
  {M}arkov decision processes. Formal Methods Syst. Des.  \textbf{50}(2-3),
  207--248 (2017). \doi{10.1007/S10703-016-0262-7}

\bibitem{SLN02}
Sarker, R.A., Liang, K.H., Newton, C.S.: A new multiobjective evolutionary
  algorithm. Eur. J. Oper. Res.  \textbf{140}(1),  12--23 (2002).
  \doi{10.1016/S0377-2217(01)00190-4}

\bibitem{Sid67}
{\v S}id{\'a}k, Z.: Rectangular confidence regions for the means of
  multivariate normal distributions. J. Am. Stat. Assoc.  \textbf{62}(318),
  626--633 (1967). \doi{10.1080/01621459.1967.10482935}

\bibitem{Sin87}
Singh, C.: Optimality conditions in multiobjective differentiable programming.
  J. Optim. Theory Appl.  \textbf{53}(1),  115--123 (1987).
  \doi{10.1007/BF00938820}

\bibitem{SXWL+22}
Song, F., Xing, H., Wang, X., Luo, S., Dai, P., Li, K.: Offloading dependent
  tasks in multi-access edge computing: A multi-objective reinforcement
  learning approach. Future Gener. Comput. Syst.  \textbf{128},  333--348
  (2022). \doi{10.1016/J.FUTURE.2021.10.013}

\bibitem{SK06}
Suman, B., Kumar, P.: A survey of simulated annealing as a tool for single and
  multiobjective optimization. J. Oper. Res. Soc.  \textbf{57}(10),  1143--1160
  (2006). \doi{10.1057/PALGRAVE.JORS.2602068}

\bibitem{Sut95}
Sutton, R.S.: Generalization in reinforcement learning: Successful examples
  using sparse coarse coding. In: Touretzky, D.S., Mozer, M., Hasselmo, M.E.
  (eds.) Advances in Neural Information Processing Systems 8 ({NIPS} 1995). pp.
  1038--1044. {MIT} Press (1995),
  \url{http://papers.nips.cc/paper/1109-generalization-in-reinforcement-learning-successful-examples-using-sparse-coarse-coding}

\bibitem{SB18}
Sutton, R.S., Barto, A.G.: Reinforcement learning - an introduction, 2nd
  Edition. {MIT} Press (2018)

\bibitem{UB13}
Ummels, M., Baier, C.: Computing quantiles in {M}arkov reward models. In:
  Pfenning, F. (ed.) 16th International Conference on Foundations of Software
  Science and Computation Structures ({FoSSaCS} 2013). Lecture Notes in
  Computer Science, vol.~7794, pp. 353--368. Springer (2013).
  \doi{10.1007/978-3-642-37075-5_23}

\bibitem{VDBID11}
Vamplew, P., Dazeley, R., Berry, A., Issabekov, R., Dekker, E.: Empirical
  evaluation methods for multiobjective reinforcement learning algorithms.
  Mach. Learn.  \textbf{84}(1-2),  51--80 (2011).
  \doi{10.1007/S10994-010-5232-5}

\bibitem{VFDB21}
Vamplew, P., Foale, C., Dazeley, R., Bignold, A.: Potential-based
  multiobjective reinforcement learning approaches to low-impact agents for
  {AI} safety. Eng. Appl. Artif. Intell.  \textbf{100},  104186 (2021).
  \doi{10.1016/J.ENGAPPAI.2021.104186}

\bibitem{VMTD15}
Vargas, D.V., Murata, J., Takano, H., Delbem, A.C.B.: General subpopulation
  framework and taming the conflict inside populations. Evol. Comput.
  \textbf{23}(1),  1--36 (2015). \doi{10.1162/EVCO_A_00118}

\bibitem{WCL19}
Wan, C., Chen, X., Liu, D.: A multi-objective-driven placement technique for
  digital microfluidic biochips. J. Circuits Syst. Comput.  \textbf{28}(5),
  1950076:1--1950076:15 (2019). \doi{10.1142/S0218126619500762}

\bibitem{WHQLL17}
Wang, C., Hou, Y., Qiu, F., Lei, S., Liu, K.: Resilience enhancement with
  sequentially proactive operation strategies. IEEE Trans. Power Syst.
  \textbf{32}(4),  2847--2857 (2017). \doi{10.1109/TPWRS.2016.2622858}

\bibitem{WS13}
Wang, W., Sebag, M.: Hypervolume indicator and dominance reward based
  multi-objective {M}onte-{C}arlo tree search. Mach. Learn.  \textbf{92}(2-3),
  403--429 (2013). \doi{10.1007/S10994-013-5369-0}

\bibitem{WKD10}
Warnquist, H., Kvarnstr{\"{o}}m, J., Doherty, P.: Iterative bounding {LAO}. In:
  Coelho, H., Studer, R., Wooldridge, M.J. (eds.) 19th European Conference on
  Artificial Intelligence ({ECAI} 2010). Frontiers in Artificial Intelligence
  and Applications, vol.~215, pp. 341--346. {IOS} Press (2010).
  \doi{10.3233/978-1-60750-606-5-341}

\bibitem{WWD14}
Wiering, M.A., Withagen, M., Drugan, M.M.: Model-based multi-objective
  reinforcement learning. In: 2014 {IEEE} Symposium on Adaptive Dynamic
  Programming and Reinforcement Learning ({ADPRL} 2014). pp.~1--6. {IEEE}
  (2014). \doi{10.1109/ADPRL.2014.7010622}

\bibitem{WZM15}
Wray, K.H., Zilberstein, S., Mouaddib, A.I.: Multi-objective {MDP}s with
  conditional lexicographic reward preferences. In: Bonet, B., Koenig, S.
  (eds.) 29th {AAAI} Conference on Artificial Intelligence ({AAAI} 2015). pp.
  3418--3424. {AAAI} Press (2015). \doi{10.1609/AAAI.V29I1.9647}

\bibitem{YNIT19}
Yamaguchi, T., Nagahama, S., Ichikawa, Y., Takadama, K.: Model-based
  multi-objective reinforcement learning with unknown weights. In: Yamamoto,
  S., Mori, H. (eds.) Thematic Area on Human Interface and the Management of
  Information ({HIMI} 2019), part of the 21st International Conference on
  Human-Computer Interaction ({HCII} 2019). Lecture Notes in Computer Science,
  vol. 11570, pp. 311--321. Springer (2019). \doi{10.1007/978-3-030-22649-7_25}

\bibitem{YSN19}
Yang, R., Sun, X., Narasimhan, K.: A generalized algorithm for multi-objective
  reinforcement learning and policy adaptation. In: Wallach, H.M., Larochelle,
  H., Beygelzimer, A., d'Alch{\'{e}} Buc, F., Fox, E.B., Garnett, R. (eds.)
  33rd Annual Conference on Advances in Neural Information Processing Systems
  ({NeurIPS} 2019). pp. 14610--14621 (2019),
  \url{https://proceedings.neurips.cc/paper/2019/hash/4a46fbfca3f1465a27b210f4bdfe6ab3-Abstract.html}

\bibitem{YS02}
Younes, H.L.S., Simmons, R.G.: Probabilistic verification of discrete event
  systems using acceptance sampling. In: Brinksma, E., Larsen, K.G. (eds.) 14th
  International Conference on Computer Aided Verification ({CAV} 2002). Lecture
  Notes in Computer Science, vol.~2404, pp. 223--235. Springer (2002).
  \doi{10.1007/3-540-45657-0_17}

\bibitem{ZTJ15}
Zhang, X., Tian, Y., Jin, Y.: A knee point-driven evolutionary algorithm for
  many-objective optimization. {IEEE} Trans. Evol. Comput.  \textbf{19}(6),
  761--776 (2015). \doi{10.1109/TEVC.2014.2378512}

\bibitem{ZLXYW22}
Zhu, Y., Liang, S., Xue, G., Yang, R., Wu, X.: An efficient multi-objective
  optimization approach for sensor management via multi-{B}ernoulli filtering.
  {EURASIP} J. Adv. Signal Process.  \textbf{2022}(1), ~62 (2022).
  \doi{10.1186/S13634-022-00881-4}

\bibitem{ZKT08}
Zitzler, E., Knowles, J.D., Thiele, L.: Quality assessment of {P}areto set
  approximations. In: Branke, J., Deb, K., Miettinen, K., Slowinski, R. (eds.)
  Outcome of the Dagstuhl seminar on Multiobjective Optimization, Interactive
  and Evolutionary Approaches. Lecture Notes in Computer Science, vol.~5252,
  pp. 373--404. Springer (2008). \doi{10.1007/978-3-540-88908-3_14}

\bibitem{ZT98}
Zitzler, E., Thiele, L.: Multiobjective optimization using evolutionary
  algorithms -- a comparative case study. In: Eiben, A.E., B{\"{a}}ck, T.,
  Schoenauer, M., Schwefel, H.P. (eds.) 5th International Conference on
  Parallel Problem Solving from Nature ({PPSN} 1998). Lecture Notes in Computer
  Science, vol.~1498, pp. 292--304. Springer (1998). \doi{10.1007/BFB0056872}

\bibitem{ZTLFF03}
Zitzler, E., Thiele, L., Laumanns, M., Fonseca, C.M., da~Fonseca, V.G.:
  Performance assessment of multiobjective optimizers: an analysis and review.
  {IEEE} Trans. Evol. Comput.  \textbf{7}(2),  117--132 (2003).
  \doi{10.1109/TEVC.2003.810758}

\end{thebibliography}

\clearpage
\renewcommand{\theHsection}{A\arabic{section}}
\appendix

\section{Additional Explanations}
\label{ap:Explanations}

To explain the need for generating new runs in the evaluation phase of our fixed-budget algorithms (as described on page~\pageref{ref:ExplanationAvoidBias}), let us make an analogy:
Consider the experiment of tossing $100$ fair coins $n$ times, and let us repeat it $1000000$ times.
If we (i)~discard all experiments but the one with the highest rate of getting heads, and then (ii)~use the already-collected data for that experiment to get a $95\,\%$ CI based on $n$ runs for that coin's probability of showing heads, we are much more likely than $95\,\%$ to falsely conclude that the coin is biased.
We thus need to collect new samples for (ii), or apply Bonferroni correction to divide the error budget over all $1000000$ repetitions to achieve simultaneous $95\,\%$ confidence.
The former needs $n$ extra runs, while (using the Okamoto bound~\cite{Oka59}) the latter now needs $10^6\cdot n$ in each repetition, so just below $10^{11}\cdot n$ extra runs in~total.

\section{Additional Figures and Plots}

\begin{figure}[b]
    \centering
    \hfill
    \subfloat[Up to $3\cdot10^9$ runs]{
        \begin{tikzpicture}[scale=0.6]
            \begin{axis}[
                name=Axis,
                xlabel={Sampled runs},
                ylabel={Sampled strategies},
                legend pos=north west,
                legend cell align=left,
                ymin=0,
                xmin=0,
                xmax=3000000000,
                xticklabel style={/pgf/number format/fixed}
            ]
                \addplot[thick,blue,dotted,mark=none] table[x=a1,y=a2] {data/approximation-settings.dat};
                \addplot[thick,red,dashed,mark=none] table[x=b1,y=b2] {data/approximation-settings.dat};
                \addplot[thick,orange,loosely dashdotted,mark=none] table[x=c1,y=c2] {data/approximation-settings.dat};
                \addplot[mark=none] table[x=d1,y=d2] {data/approximation-settings.dat};
                \legend{$f=0.1;m=100$, $f=0.1;m=1000$, $f=0.5;m=100$, $f=0.5;m=1000$}
                \addplot[black,dashed,mark=none] coordinates {(0, 1000) (300000000, 1000) (300000000, 0)};
            \end{axis}
        \end{tikzpicture}
        \label{fig:approximation-settings:large}
    }
    \hfill
    \subfloat[Zoom in to up to $3\cdot10^8$ runs]{
        \begin{tikzpicture}[scale=0.6]
            \begin{axis}[
                name=Axis,
                xlabel={Sampled runs},
                ylabel={Sampled strategies},
                legend pos=north west,
                legend cell align=left,
                ymin=0,
                xmin=0,
                xmax=300000000,
                xticklabel style={/pgf/number format/fixed}
            ]
                \addplot[thick,blue,dotted,mark=none,restrict x to domain=0:300000000] table[x=a1,y=a2] {data/approximation-settings.dat};
                \addplot[thick,red,dashed,mark=none,restrict x to domain=0:300000000] table[x=b1,y=b2] {data/approximation-settings.dat};
                \addplot[thick,orange,loosely dashdotted,mark=none,restrict x to domain=0:300000000] table[x=c1,y=c2] {data/approximation-settings.dat};
                \addplot[mark=none,restrict x to domain=0:300000000] table[x=d1,y=d2] {data/approximation-settings.dat};
                \legend{$f=0.1;m=100$, $f=0.1;m=1000$, $f=0.5;m=100$, $f=0.5;m=1000$}
            \end{axis}
        \end{tikzpicture}
        \label{fig:approximation-settings:small}
    }
    \hfill${}$

    \caption{Sampled strategies per sampled runs in \Cref{alg:incremental} for $d = 2, \varepsilon = 0.01, \alpha = 0.1$}
    \label{fig:approximation-settings}
\end{figure}

\Cref{fig:approximation-settings} provides two additional views on the tradeoffs in \Cref{alg:incremental} (the incremental scheme) that are further zoomed out (\Cref{fig:approximation-settings:large}) and further zoomed in (\Cref{fig:approximation-settings:small}) compared to \Cref{fig:approximation-settings:middle} in \Cref{sec:incremental}.

\section{Details on Benchmark Models}
\label{ap:Models}

We provide more details on the families of benchmark models that we added over the ones collected from QComp 2023 to demonstrate that our algorithms can analyse models that PMC struggles with.

\paragraph{Breakable bottles.}
The breakable bottles model~\cite{VFDB21} represents a person that wants to move bottles from position A to position B.
However, if they carry more than one bottle at a time, with each movement there is a 10\,\% chance of dropping, and thereby breaking, one bottle.
At some point, the person will get exhausted from all of this exercise and has to stop.
There are three rewards, representing (1.) the number of broken bottles, (2.) the number of successfully transferred bottles, and (3.) the number of steps taken.

\paragraph{Deep sea treasure.}
The deep sea treasure model~\cite{VDBID11} is commonly used in RL~\cite{FANB+23,WWD14,CWG20,HLIS+19,WS13}.
In this model, a deep sea submarine searching for treasure is controlled.
Larger treasures are located deeper underwater, which introduces trade-off with fuel consumption.
We use two variants of this model: a deterministic and a probabilistic version.
The original deep sea treasure model is deterministic, which allows for a straightforward computation of the correct Pareto front~\cite{VDBID11}.
To make it more interesting, we use a larger environment than the original model, which makes it challenging for PMC.
The probabilistic version introduces a probabilistic chance of the submarine imploding when it submerges too quickly, thus requiring greater patience.

\paragraph{Fruit tree model.}
The fruit tree model~\cite{YSN19} is a perfect binary tree with each leaf handing out rewards corresponding to nutrients of the fruit on that leaf.
To obtain non-trivial Pareto fronts, the reward vector with 2 objectives is chosen such that each leaf corresponds to a discrete point in the positive quadrant of the quadratic equation $1 - x^2$.
For 6 objectives, each leaf corresponds to a discrete point on the all-positive section of a 5-sphere.
Since each state corresponds to two new states in the next depth, the state space grows exponentially.
Since the state space explodes rapidly, these models are computationally infeasible for PMC.
However, the length of the paths scales linearly with depth, making these models suitable for SMC.

\paragraph{Puddle world racetracks.}
We finally include racetrack benchmarks~\cite{Gar73}, which are also used in AI literature~\cite{BBS95,WKD10,ADLR22,FGR22} and have been adopted by the PMC community as well~\cite{BCGG+20,GHHKS23,CEHH+21}.
In these models, the driver must select a direction in which they would like to accelerate at every timestamp.
The car then has a 90\,\% chance of accelerating in that direction.
The driver aims to optimise the probability of reaching the finish line without leaving the track.
To make them multi-objective, we use the classic puddle world~\cite{BM94} reward scheme, which is widely used in RL literature~\cite{Sut95,NNVN+20,JKP21,IGIT03,LTZCX24}.
In the puddle world, the agent aims to finish as soon as possible but would like to avoid the puddles, which are the off-track areas of the map.
Since the path to the finish is shorter when driving off-track, there is a trade-off between the puddle punishment and the time to reach the finish line.

\section{Additional Experimental Results}
\label{ap:AdditionalResults}

We provide some more plots to illustrate our findings.
For each model, we provide a plot analogous to \Cref{fig:seed} of the three seeds, using FSB, the simple heuristic, and configuration 3 since we found in \Cref{sec:Experiments} that those performed best overall.
The results are shown in \Cref{app:fig:deep_sea,app:fig:deep_sea_probabilistic,app:fig:energy_aware_scheduling,app:fig:fruit_tree_2,app:fig:mars_rover,app:fig:puddle_world_barto_big,app:fig:puddle_world_barto_small,app:fig:puddle_world_hansen,app:fig:puddle_world_ring,app:fig:puddle_world_river,app:fig:puddle_world_tiny}; the true Pareto front is shown in black in case it is known.

The true Pareto fronts are known by construction for the fruit-tree, deep-sea (both variants) and breakable-bottles models.
For the large fruit-tree there are too many Pareto-optimal points to feasibly compute the multiplicative epsilon.

In the breakable-bottles model, a Pareto-optimal strategy is to only walk with one bottle from A to B and back.
The probability of finding this strategy, or a very close one, is low.
Hence, this model serves as an example of when our approach is not suitable (but we did not select it for this purpose a priori).

\begin{figure}[tbp]
    \centering
	\subfloat[\texttt{.5-100}]{

    }%
    \caption{\texttt{deep\_sea}, FSB, Simple, C3}
    \label{app:fig:deep_sea}
\end{figure}

\begin{figure}[tbp]
    \centering
	\subfloat[\texttt{.5-100}]{
        %
    }%
    \caption{\texttt{deep\_sea\_probabilistic}, FSB, Simple, C3}
    \label{app:fig:deep_sea_probabilistic}
\end{figure}

\begin{figure}[tbp]
    \centering
	\subfloat[\texttt{.3-2}]{
        %
    }%
    \caption{\texttt{energy\_aware\_scheduling}, FSB, Simple, C3}
    \label{app:fig:energy_aware_scheduling}
\end{figure}

\begin{figure}[tbp]
    \centering
	\subfloat[\texttt{.5}]{
        %
    }%
    \caption{\texttt{fruit\_tree\_2}, FSB, Simple, C3}
    \label{app:fig:fruit_tree_2}
\end{figure}

\begin{figure}[tbp]
    \centering
	\subfloat[\texttt{.10-1}]{
        %
    }%
    \caption{\texttt{mars\_rover}, FSB, Simple, C3}
    \label{app:fig:mars_rover}
\end{figure}

\begin{figure}[tbp]
    \centering
	\subfloat[\texttt{.30}]{
        %
    }%
    \caption{puddle\_world\_barto\_big, FSB, Simple, C3}
    \label{app:fig:puddle_world_barto_big}
\end{figure}

\begin{figure}[tbp]
    \centering
	\subfloat[\texttt{.50}]{
        %
    }%
    \caption{puddle\_world\_barto\_small, FSB, Simple, C3}
    \label{app:fig:puddle_world_barto_small}
\end{figure}

\begin{figure}[tbp]
    \centering
	\subfloat[\texttt{.80}]{
        %
    }%
    \caption{puddle\_world\_hansen, FSB, Simple, C3}
    \label{app:fig:puddle_world_hansen}
\end{figure}

\begin{figure}[tbp]
    \centering
	\subfloat[\texttt{.60}]{
        %
    }%
    \caption{puddle\_world\_ring, FSB, Simple, C3}
    \label{app:fig:puddle_world_ring}
\end{figure}

\begin{figure}[tbp]
    \centering
	\subfloat[\texttt{.30}]{
        %
    }%
    \caption{puddle\_world\_river, FSB, Simple, C3}
    \label{app:fig:puddle_world_river}
\end{figure}

\begin{figure}[tbp]
    \centering
	\subfloat[\texttt{.5}]{
        %
    }%
    \caption{\texttt{puddle\_world\_tiny}, FSB, Simple, C3}
    \label{app:fig:puddle_world_tiny}
\end{figure}

\begin{figure}[t]
    \begin{minipage}[t]{0.5\textwidth}
        \centering
        %
        \caption{Configurations}\label{fig:app-conf}
    \end{minipage}%
    \begin{minipage}[t]{0.5\textwidth}
        \centering
        %
        \caption{Algorithms}\label{fig:app-algs}
    \end{minipage}%
\end{figure}

\end{document}